\documentclass{article}
\usepackage{enumitem,verbatim}
\usepackage{microtype}
\usepackage[margin=70pt]{geometry}
\usepackage{amsmath,amssymb,amsthm}
\usepackage{graphicx}

\usepackage[
    backend=biber,natbib,
    style=alphabetic
]{biblatex}
\renewbibmacro{in:}{}

\newtheorem{thm}{Theorem}[section]

\newtheorem{lem}[thm]{Lemma}

\newtheorem{prop}[thm]{Proposition}

\theoremstyle{definition}

\newtheorem{defn}[thm]{Definition}

\DeclareMathOperator*{\argmax}{arg\,max}

\title{A Model of Market Making and Price Impact}
\author{Angad Singh\footnote{Department of Mathematics, California Institute of Technology. Email: assingh@caltech.edu}}

\begin{document}
\maketitle

\begin{abstract}
Traders constantly consider the price impact associated with changing their positions. This paper seeks to understand how price impact emerges from the quoting strategies of market makers. To this end, market making is modeled as a dynamic auction using the mathematical framework of Stochastic Differential Games. In Nash Equilibrium, the market makers' quoting strategies generate a price impact function that is of the same form as the celebrated Almgren-Chriss model. The key insight is that price impact is the mechanism through which market makers earn profits while matching their books. As such, price impact is an essential feature of markets where flow is intermediated by market makers. 
\end{abstract}

\section{Introduction}
Execution costs are an important issue to consider when implementing trading strategies. In electronic markets traders typically use sophisticated algorithms to sequentially place small orders that minimize trading costs. Many trading costs come in the form of market impact, which can be understood by looking at the following stylized snapshot of a limit order book (LOB). 
\begin{figure}[h!]
\centering
\includegraphics[scale=0.5]{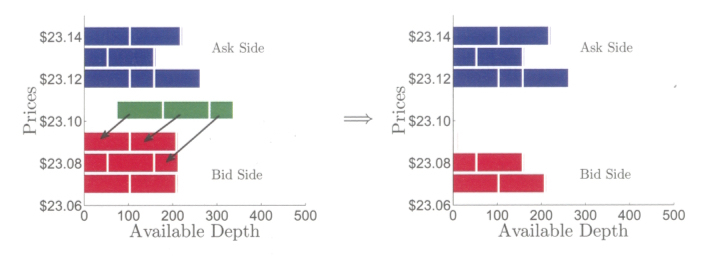}
\caption{A stylized limit order book \cite{carteaAlgorithmicHighFrequencyTrading2015}}
\end{figure}

In Figure 1, the blue and red orders are outstanding limit orders, which correspond to outstanding quotes to trade by market participants. The ask side consists of quotes to sell, and the bid side consists of quotes to buy. The green order is a market order, which in this case is to sell. A market order specifies a quantity to buy or sell and is executed against the best outstanding quotes. In general, limit order are thought of as providing liquidity, as they make it possible for others to trade. They are usually placed by market makers. Market orders are thought of as taking liquidity, as their placement requires someone else to also trade. 

The market order in this figure suffers from two important forms of market impact. Firstly, because the order is somewhat large, it cannot all be executed at the best bid price (at the touch), and so the average price received per share in the order is strictly smaller than \$23.09. This is called walking the book, or temporary price impact; the larger the order the worse the price per share received. 

The second form of market impact is more subtle. Note that on the right side of Figure 1, which is the LOB immediately after the market order executes, the shape of the book is different. More specifically, the bid-ask spread is wider (the difference between the best available buy and sell prices) and the midprice is lower (the average of the best buy and sell prices). Typically market makers will change their limit orders right after this, and the shape of the book will further change. On average though, the midprice will be lower/higher after a (large) sell/buy market order. This is known as permanent price impact.

A common mathematical framework to think about such price impact is known as the Almgren-Chriss model. A trader is assumed to sell at a continuous trading rate $q_t$ over a time interval $[0, T]$. A positive (negative) trading rate corresponds to selling (buying) $q_tdt$ shares in the time instant $[t, t+dt]$. The trader's transaction prices are then given by 
\begin{equation*}
p_t = D_t - \Gamma\int_0^t q_sds - \Lambda q_t.
\end{equation*} 
Here $\Gamma > 0$ corresponds to permanent price impact, as it is the downward pressure induced on the price by the accumulation of sell orders on $[0,t)$. $\Lambda > 0$ corresponds to temporary price impact, as it is the downward pressure induced on the price by the size of the current order. The process $\{D_t\}$ can be thought of as the fundamental value, or unaffected midprice. It is what the midprice would have been in the absence of the trader's selling, and it is often modeled as an arithmetic Brownian Motion. 

This paper seeks to understand how a price impact function such as the one above would emerge from the quoting strategies of market makers, as they are the ones placing the orders against which the trader above is executing. The paper takes as given an exogenous stream of market orders hitting the limit order book, and then consider an auction between market makers to clear this flow. Quoting between the market makers is modeled as submitting demand functions to trade, i.e. announcing a set of acceptable price quantity pairs. 

The paper presents a static one period model (Section 2) as well as a dynamic continuous time model (Section 3). The one period model transparently highlights the tradeoffs faced by the market makers, as well as the mechanism linking price impact to their profits. The dynamic model allows for a more realistic model where liquidity traders' flow mean reverts in the long run. This means that market makers are able to match their books over time, and on average have flat exposures. In the short run, however, they have one-sided exposures, and thus are exposed to fluctuations in fundamental value. Price impact is the mechanism through which market makers ensure that they don't lose money while taking on these short term exposures. 

All in all, the models in this paper deliver the following three insights regarding market making and price impact. Firstly, in order to be profitable, a market maker should make sure that the following heuristic is true for her trades:
\begin{equation*}
sign(trade)(P^{trade} - P^{after}) < 0.
\end{equation*}
When a market maker buys, $sign(trade) > 0$, she should make sure to do so at a price slightly lower than what (she rationally thinks) the midprice will be right after the trade, and vice versa. The reason is that the market maker is hoping to make the opposite trade and sell right after, and in order for this to be profitable the inequality above must hold (on average). This is discussed in more detail when presenting equation (28) below.

Secondly, in order to implement a quoting strategy that achieves the inequality above, the market maker should proceed as follows. At any given moment, the market maker can look at her outstanding quotes and compute the function $P(q)$, which is the trading price she will get if she ends up trading a quantity of $q$. P(q) should reflect expectations of future prices conditional on the fact that a trade of size $q$ is occurring, and this should hold for every $q$. This is necessary to make sure the market maker makes money (on average) regardless of which of her quotes end up executing today. A quoting strategy achieving this can be implemented only if orders have market impact, i.e. if transaction prices depends on order sizes. This is discussed in more detail when presenting Proposition 5.4 below.  

Finally, given this theoretical foundation of price impact from the perspective of market makers, we can provide insight on two common execution cost measures used in practice. The first is effective cost, measured heuristically when selling by 
\begin{equation*}
P^{before} - P^{trade}.
\end{equation*}
The second is realized cost, measured heuristically when selling by
\begin{equation*}
P^{after} - P^{trade}.
\end{equation*}
The first measure looks at the entire price impact of the trade, and the second measure looks only at the subsequent price reversal. In general the first measure is larger due to the long lasting impact of the current trade on the price \cite{pedersenEfficientlyInefficientHow2015}. The theory in this paper supports this, as in the model the second measure differs from the first because of market makers expectations of when they will eventually be able to match their book. Furthermore, the theory suggests that the measures also differ because market making is not a perfectly competitive activity. It's certainly true that there are barriers to entry (e.g. technological and regulatory) in the market making industry, and it would be an interesting empirical exercise to the test the hypothesis that this drives the wedge between the two measures. 

\subsection{Related Literature}
The contents of this paper relate to three strands of existing academic literature. The first is the theoretical economics literature on market microstructure. The second is the mathematical finance literature on optimal execution. The third is the finance literature on asset pricing. 

The literature on market microstructure is fairly classical and considers various one-period equilibrium auction models. Thorough reviews are given in \cite{oharaMarketMicrostructureTheory2011a} and \cite{foucaultMarketLiquidityTheory2013}. The demand schedule auction used in this thesis was first introduced in \cite{wilsonAuctionsShares1979}. \cite{kyleInformedSpeculationImperfect1989} uses the demand schedule auction to study price impact as a consequence of adverse selection. The contents of section 2 are similar to these papers, though the specific theorems proven are new.  One major way in which section 2 (and the rest of the paper) deviates from \cite{kyleInformedSpeculationImperfect1989} is that price impact is not a consequence of adverse selection. Instead, price impact compensates market makers for absorbing short term imbalances. This is in line with the idea of liquidity as the price of immediacy, as introduced in \cite{grossmanLiquidityMarketStructure1988}.

The literature on optimal execution considers dynamic partial equilibrium problems for an individual agent trading against a price impact function. This literature was pioneered in \cite{almgrenOptimalExecutionPortfolio2001} and has since been extended in a variety of directions, as reviewed in \cite{carteaAlgorithmicHighFrequencyTrading2015} and \cite{gueantFinancialMathematicsMarket2016}. The optimal response problems for the model in section 3 is very much analogous to the optimization problems in this literature. However, this literature always takes the market price impact function as exogenously given, whereas in this paper the price impact function is endogenous. 

Asset pricing is the standard framework used in the academic finance community to model market prices. Thorough introductions to the subject are given in \cite{cochraneAssetPricing2005} and \cite{backAssetPricingPortfolio2017}. This literature does not explicitly model auctions, as the market microstructure literature does, but instead uses a general equilibrium approach. Furthermore, trading costs are usually ignored in these models. \cite{garleanuDynamicPortfolioChoice2016} and \cite{bouchardEquilibriumReturnsTransaction2018} do consider trading costs in asset pricing settings, but these costs are exogenously specified. Also, the trading costs in these models are not in the form of market price impact, but instead are transaction costs paid on top of market prices. Many of the modeling techniques in this paper are inspired by these papers, and also by the continuous time CARA-Normal frameworks used in \cite{campbellSmartMoneyNoise1993} and \cite{barberisXCAPMExtrapolativeCapital2015}  

\section{The One Period Model}
This section formulates and discusses some one period models of auctions for shares\footnote{For our purposes a share is simply a unit of an infinitely divisible good.}. The type of auction considered is typically referred to as a conditional uniform price auction, or as the demand schedule game. At the beginning of the auction, each agent submits a demand schedule $D: \mathbb{R} \to \mathbb{R}$, which is a commitment to purchase $D(p)$ shares if the price per share is $p$. The auctioneer then observes all the schedules, aggregates them into a total demand schedule, and chooses a price such that the auction clears. Agents receive shares based on their individual demand schedules, and finally they derive utility from the payoff associated to each share. 

This type of auction is interesting because it's a good way to model agents trading through a limit order book. At a very high level, the decision to place market and limit orders amounts to deciding how aggressively to trade. A trader with a great sense of urgency might submit a large market order, which executes immediately, but is subject to walking the book. On the other hand, a more patient trader might scatter a few limit orders throughout the book, which may or may not execute, but will do so at a good price if they do. Finally, note then when an order executes it must, by definition, execute against another order placed by another agent. Thus agents are not just deciding how aggressively to push their own trading agendas, but also how aggressively to absorb the order flow of other traders.

Returning to the auction, note that selecting a demand schedule also amounts to deciding how aggressively to trade. A vertical demand schedule is equivalent to submitting a market order, as it demands a fixed quantity regardless of the price received. Tilting the slope of the demand schedule downwards is much like placing limit orders. It gives the agent the opportunity to trade at a good price, but this is not a guarantee because there must be another agent willing to accept the other side of this trade. Finally, from a game theoretic point of view, when an agent takes others' schedules as given she is a monopsonist facing a supply curve. The suppliers are of course the other agents in the game, so forming an optimal response involves pushing ones own trading agenda while also deciding how much of others' order flow to absorb. 

Hence the decision-making process of an agent in the demand schedule game captures many features of trading on a limit order book.  An interesting property of the demand schedule game is that it features multiple equilibria when there is no uncertainty in the quantity to be cleared at the auction. From our point of view this is a realistic feature, as the quantity to be cleared on the exchange at any instant is uncertain. This first section below formulates and discusses this result. The second section presents a version of the model where agents have initial inventories, thereby formulating it as a Bayesian game. This version of the model connects better with the dynamic models to follow. 

\subsection{The Model}
There are $N \geq 2$ market makers bidding for a total of $S$ outstanding shares. Each share will, after the auction, provide a random payoff of $\tilde{\mu} \sim \mathcal{N}(\mu, \sigma^2)$. Market makers are assumed to have CARA preferences over payoffs, with risk aversion parameter $\gamma > 0$. Thus if a market maker purchases $q$ shares in the auction for a price of $p$ per share, her expected utility is $\mathbb{E}[-e^{-\gamma q(\tilde{\mu} - p)}]$. We will assume that market makers can only submit affine and decreasing demand schedules\footnote{The affine assumption is not strictly necessary, as an optimal response to a profile of affine schedules is also given by an affine schedule.}. So a demand schedule is characterized by a pair $(a, b) \in \mathbb{R} \times (0, \infty)$, corresponding to the commitment to purchase $a-bp$ shares when the price per share is $p$.\footnote{$b = \infty$ corresponds to placing an order that specifies a price but no quantity. Such an order is not possible on an exchange, and thus is not allowed here. Technically we should allow $b = 0$, since it corresponds to placing a market order. However, we only consider symmetric equilibria below, and $b=0$ can trivially never be such an equilibrium. Thus for the sake of exposition $b = 0$ is omitted from the outset.} 

In addition to the $N$ market makers there are liquidity traders who have a random perfectly inelastic demand for $-\tilde{u}$ shares. Thus the total number of shares the $N$ market makers must clear in the auction is $S + \tilde{u}$. By adjusting the mean of $\tilde{u}$ we can simply assume that the total number of shares being auctioned is random and equal to $\tilde{u}$. We assume that $\tilde{u}$ is independent of $\tilde{\mu}$, but make no other assumptions about its distribution.

Given demand schedules $(a^1, b^1), \cdots, (a^N, b^N)$, the auction price $p$ and the shares $q^1, \cdots, q^N$ bought by each market maker are implicitly given by 
\begin{align}
a^n - b^np = q^n \\
q^1 + \cdots + q^N = \tilde{u}.
\end{align}
This describes a game in normal form, and we now proceed to study its Nash equilibria. We will focus only on symmetric equilibria, that is on equilibria where all market makers submit the same demand schedule. Note that in a symmetric equilibrium all market makers purchase the same number of shares in the auction, i.e. $q^1 = \cdots = q^N = \frac{\tilde{u}}{N}$. 

\begin{thm}
Fix exogenous parameters $N \geq 2, \gamma,\sigma>0$ and $\mu \in \mathbb{R}$. \\
If $\tilde{u}$ is degenerate, i.e. $\tilde{u} = u \in \mathbb{R}$ a.s., then there is a one-to-one correspondence between symmetric equilibria and $\lambda > 0$. In equilibrium the price is 
\begin{equation}
p = \mu - \frac{\gamma\sigma^2}{N}u - \frac{\lambda}{N} u. 
\end{equation}
If $\tilde{u}$ is non-degenerate and $N \geq 3$, then there is a unique symmetric equilibrium with price
\begin{equation}
p = \mu - \frac{\gamma\sigma^2}{N}\tilde{u} - \frac{\gamma \sigma^2}{N(N-2)} \tilde{u}.
\end{equation}
If $\tilde{u}$ is non-degenerate and $N = 2$ then a symmetric equilibrium does not exist. 
\end{thm}

\begin{proof}
Fix a strategy $(a, b) \in \mathbb{R} \times (0, \infty)$ to be played by all but one agent, and consider the optimal response problem faced by the remaining agent. If this agent plays strategy $(\alpha, \beta) \in \mathbb{R} \times (0, \infty)$, then her expected utility will be $\mathbb{E}[-e^{-\gamma q(\tilde{\mu} - p)}]$ with $p$ and $q$ given implicitly by
\begin{align}
q &= \alpha - \beta p \\
\frac{\tilde{u} - q}{N-1} &= a-bp.
\end{align}
Rearranging we obtain
\begin{align}
p &= F - \lambda\tilde{u} + \lambda q \\
q &= \bigg[\frac{\alpha}{1 + \lambda\beta} - \frac{\beta}{1+\lambda\beta}F\bigg] + \frac{\lambda\beta}{1+\lambda\beta}\tilde{u},
\end{align}
where $F := \frac{a}{b}$ and $\lambda := \frac{1}{b(N-1)}$.

We note a few things about these equations. Firstly, the parameters $F$ and $\lambda$ characterize the symmetric profile given by $(a,b)$ and are independent of the remaining agent's demand schedule. Secondly, equation (7) tells us that the price the remaining agent receives is uniquely determined by the number of shares she receives. Thus the agent is indifferent between schedules that lead to the same quantity of shares. Thirdly, from equation (8) we see that choosing $(\alpha, \beta) \in \mathbb{R} \times (0, \infty)$ amounts to choosing to receive the quantity $q \in V$ where
\begin{equation}
V := \{A + B\tilde{u}: (A, B) \in \mathbb{R} \times (0, 1)\}.
\end{equation}
Fourthly, if the agent submits the schedule $(a,b)$ then she receives the quantity $q = \frac{\tilde{u}}{N}$.

All this goes to show that in forming an optimal response, the remaining agent can maximize directly over $q \in V$, with the price given by (7), and in equilibrium the maximum must be attained at $\frac{\tilde{u}}{N}$. Thus symmetric equilibria correspond to $F \in \mathbb{R}$ and $\lambda > 0$ such that 
\begin{equation}
\frac{\tilde{u}}{N} \in \argmax_{q \in V} \mathbb{E}\bigg[-e^{-\gamma q\big(\tilde{\mu} - F + \lambda\tilde{u}- \lambda q\big)}\bigg],
\end{equation}
and in equilibrium the price is 
\begin{equation}
p = F - \frac{N-1}{N}\lambda\tilde{u}.
\end{equation}

Now, note that $\tilde{u}$ is independent of $\tilde{\mu}$ and every $q \in V$ is $\tilde{u}$ measurable. Since the MGF of $\tilde{\mu}$ is $\mathbb{E}[e^{t\tilde{\mu}}] = e^{\mu t + \frac{\sigma^2}{2}t^2}$, we can compute the expectation in (10) for any $q \in V$ as
\begin{align}
\mathbb{E}\Bigg[-e^{-\gamma q\Big(\mu - F + \lambda\tilde{u} - (\lambda + \frac{\gamma \sigma^2}{2}) q\Big)}\Bigg].
\end{align}

For any $F \in \mathbb{R}$ and $\lambda > 0$, the function of $q$ in (12) is strictly concave over the convex set $V$. Thus the argmax in (10) consists of at most one point. Furthermore, for any $F \in \mathbb{R}$, $\lambda > 0$ and $u \in \mathbb{R}$, the much relaxed problem of maximizing the exponent in (12)
\begin{equation*}
\max_{q \in \mathbb{R}} q(\mu - F + \lambda u - \Big(\lambda + \frac{\gamma \sigma^2}{2}\Big)q) 
\end{equation*}
has unique solution $\hat{q} = \frac{\mu - F}{2\lambda + \gamma \sigma^2} + \frac{\lambda}{2\lambda + \gamma \sigma^2}u$.

It follows that (10) holds if and only if
\begin{equation}
\frac{\tilde{u}}{N} = \frac{\mu - F}{2\lambda + \gamma \sigma^2} + \frac{\lambda}{2\lambda + \gamma \sigma^2}\tilde{u}.
\end{equation}
Thus symmetric equilibria correspond to $F \in \mathbb{R}$ and $\lambda > 0$ satisfying (13), and in equilibrium the price is given by (11). 

Now, if $\tilde{u}$ is degenerate and equal to $u \in \mathbb{R}$ a.s., then (13) simply reads 
\begin{equation}
F = \mu - \frac{\gamma\sigma^2}{N}u + \frac{N-2}{N}\lambda u.
\end{equation}
Thus there a one-to-one correspondence between symmetric equilibria and $\lambda > 0$, with $F$ given by (14). The first statement in the theorem follows by using (14) to substitute for $F$ in (11).

Next suppose that $\tilde{u}$ is non-degenerate. Then (14) holds if and only if $F = \mu$ and $\lambda > 0$ satisfies
\begin{equation}
(N-2)\lambda = \gamma\sigma^2.
\end{equation}
If $N = 2$ then no $\lambda > 0$ can satisfy (15), so a symmetric equilibrium does not exist, proving the last statement in the theorem. If $N \geq 3$ then the unique $\lambda > 0$ satisfying (15) is $\lambda = \frac{\gamma\sigma^2}{N-2}$, and so there is a unique symmetric equilibrium. The second statement in the theorem now follows by plugging in $F = \mu$ and $\lambda = \frac{\gamma\sigma^2}{N-2}$ in (11). 
\end{proof}

\subsubsection{Discussion}
The idea in the proof is to characterize symmetric profiles in terms of the parameters $F$ and $\lambda$ of the induced the optimal response problem. The optimal response problem is to choose an expected utility maximizing quantity on the linear supply curve (7), which has intercept $F - \lambda\tilde{u}$ and slope $\lambda > 0$. This is a concave maximization problem with a unique solution. The symmetric equilibrium condition on $F$ and $\lambda$ is that this solution is $\frac{\tilde{u}}{N}$. In the non-degenerate case this condition uniquely determines $F$ and $\lambda$ whereas in the degenerate case it only specifies $F$ as a function of $\lambda$.

What happens in the degenerate case is that the market makers are able to form an agreement to misprice the asset and then take equal shares of the profit. Consider for example the case when $\tilde{u} = 1$, so a unit share is being cleared at the auction. Then (3) states that the equilibrium price can take on any value below $\mu - \frac{\gamma\sigma^2}{N}$, which is the price the asset would trade at in a competitive equilibrium. Thus we see that the asset is being priced relatively low, and since each market maker takes $\frac{1}{N}$ shares, they split the profits equally. 

Since the game is non-cooperative, in order to form an agreement the market makers must have a way to prevent others from taking more than an equal share of the profits. The key point is that all market makers submit entire demand schedules, which specify what the price must be contingent on the quantity the market maker receives. So if one market maker were to take more than $\frac{1}{N}$ shares, some other market makers would receive less than $\frac{1}{N}$ shares, and this would cause the price to move, thus dissuading any one market maker from trying to take extra shares in the first place. 

The amount by which the price would move if a market maker took more then $\frac{1}{N}$ shares is governed by the parameter $\lambda$, which corresponds to the quantity elasticity $b$ of the equilibrium demand schedule. The lower the equilibrium price $p$, the more of an incentive a market maker has to acquire more than $\frac{1}{N}$ shares, and thus the higher $\lambda$ needs to be to prevent the market maker from doing so. (3) says exactly that low equilibrium prices correspond to high values of $\lambda$.

The problem in the degenerate case is that market makers suffer no cost from being quantity elastic, since there will be no surprise trades in equilibrium. Hence $\lambda$ can take on any positive value in equilibrium. In the non-degenerate case, market makers suffer costs from being quantity elastic in equilibrium, because there is uncertainty in the quantity to be cleared. These costs manifest in how the parameter $\lambda$ effects the uncertainty of equilibrium prices. Since market makers care about the uncertainty of prices, this pins down the unique equilibrium value of $\lambda$ as $\frac{\gamma\sigma^2}{N-2}$.

The coefficient of $\tilde{u}$ in (4) is $\frac{N-1}{N(N-2)}\gamma\sigma^2$, which is the price impact of the liquidity traders' order. If the liquidity traders sell $\epsilon$ more shares, so the realization of $\tilde{u}$ is $\epsilon$ higher, then the equilibrium price is $\frac{N-1}{N(N-2)}\gamma\sigma^2\epsilon$ lower. Thus orders walk the book: the larger an order, the lower the transaction price.

The decomposition of price impact into the two terms is motivated by considering the competitive limit as $N \to \infty$ and $\frac{\gamma}{N}$ is held fixed. In the limit the second term vanishes and only the first remains. Thus $\mu - \frac{\gamma\sigma^2}{N}\tilde{u}$ is the competitive benchmark, and the second term is the deviation due to imperfect competition. As in the competitive case, the term $\gamma\sigma^2\frac{\tilde{u}}{N}$ is the risk compensation each market maker requires to take the equilibrium exposure of $\frac{\tilde{u}}{N}$. 

The interpretation is that price impact arises for two reasons in the model. Firstly to make sure market makers are appropriately compensated for bearing risk, and secondly because market makers have market power. The first reason persists even in the competitive limit, and as a result price impact does not vanish in the limit. This will be a recurring theme throughout the paper. 

\subsection{The Model with Inventories}
The continuous time model considered in the next section essentially consists of the auction above at each instant, with the addition of certain state variables that need to be carried from instant to instant. The state variables are the existing inventories of shares that market makers have accumulated from trading in the past. This section introduces these state variables in a static setting as types, thus generalizing the model above to a Bayesian game. 

In addition to forming a tighter connection with the continuous time model to follow, the rephrased model in this section has two other appealing features. Firstly, in the previous subsection the total number of outstanding shares played no distinct role\footnote{In the previous section, the total number of outstanding shares was absorbed into the mean of $\tilde{u}$.} from the liquidity traders' order. This is perhaps counterintuitive, as the liquidity traders' order should have price impact, whereas the total number of outstanding shares should be a fixed component of the price. In this section the total number of outstanding shares will show up as a fixed component of the price. Secondly, in the previous section all market makers purchased the same number of shares in equilibrium. In this section their purchases will be heterogenous.

The model is exactly as before except that each of the $N$ market makers starts out with an existing inventory of $X^n \in \mathbb{R}$ shares. Thus if market maker $n$ purchases $q^n$ shares in the auction for a price of $p$ per share, then her expected utility is 
\begin{equation}
\mathbb{E}\Big[-e^{-\gamma\big((X^n + q^n)\tilde{\mu} - pq^n\big)}\Big].
\end{equation}
The liquidity traders start out with zero shares, and the total number of outstanding shares is $S$, so $\sum_{n = 1}^N X^n = S$.

More formally, we work on a probability space with a single objective probability measure. There are $N+2$ real-valued random variables defined on this probability space: $\tilde{\mu}$, $\tilde{u}$, and $X^1, \cdots, X^N$. There are exogenous constants $\mu, S \in \mathbb{R}$ and $\sigma > 0$ such that $\tilde{\mu} \sim \mathcal{N}(\mu, \sigma^2)$ and $\sum_{n = 1}^N X^n = S$. Furthermore, $\tilde{\mu}$ and $\tilde{u}$ are independent of each other as well as $X^1, \cdots, X^N$.

The type (or private information) of market maker $n$ is $X^n$. A strategy is a measurable function mapping the realization of a market maker's type to a choice of demand schedule. As before demand schedules are restricted to be affine and strictly decreasing, so a strategy for market maker $n$ is a measurable mapping $(a^n, b^n): \mathbb{R} \to \mathbb{R} \times (0, \infty)$, $X \mapsto (a^n(X), b^n(X))$. Given a strategy profile and a realization of $(X^1, \cdots, X^N)$, prices and quantities are determined from (1) and (2) with $a^n = a^n(X^n)$ and $b^n = b^n(X^n)$. 

This completes the description of the model as a Bayesian game. We will be interested in identifying Bayesian Nash equilibria\footnote{What we call a Bayesian Nash equilibrium is sometimes called a strong Bayesian Nash equilibrium. We require market maker $n$ to choose a strategy that maximizes (16) \emph{conditional} on $X^n$ for \emph{every realization} of $X^n$. This is in contrast to the weaker requirement of choosing a strategy that just maximizes (16), which \emph{averages over realizations} of $X^n$.} in this game, but we will focus on equilibria that have a very specific structure. 
\begin{defn}
A strategy $s: \mathbb{R} \to \mathbb{R} \times (0, \infty)$ is called \textbf{linear} if there exist constants $a, \xi \in \mathbb{R}$ and $b \in (0, \infty)$ such that 
\begin{equation*}
s(X) = \big(aX + \xi, b\big) \ \ \ \forall X \in \mathbb{R}.
\end{equation*}
\end{defn}
Our focus will be on linear symmetric equilibria, that is on equilibria where all market makers play the same linear strategy.

\begin{thm}
Fix exogenous parameters $N \geq 3$, $\gamma, \sigma > 0$ and $\mu, S \in \mathbb{R}$. 

If $\tilde{u}$ is non-degenerate then there is a unique linear symmetric equilibrium. In equilibrium, the price is
\begin{equation}
p = \mu - \frac{\gamma\sigma^2}{N}S - \frac{N-1}{N(N-2)}\gamma\sigma^2\tilde{u} 
\end{equation}
and the quantities purchased by each market maker are
\begin{equation}
q^n = -\frac{N-2}{N-1}(X^n - \frac{S}{N}) + \frac{\tilde{u}}{N}.
\end{equation}
\end{thm}

\begin{proof}
Fix a linear strategy given by $a, \xi \in \mathbb{R}$ and $b \in (0, \infty)$ to be played by all but one agent, and consider the optimal response problem faced by the remaining agent. Suppose the remaining agent plays the strategy $\mathbb{R} \to \mathbb{R} \times (0,\infty), x \mapsto (\alpha(x), \beta(x))$. If the agent's initial inventory is $X$, then her expected utility is $\mathbb{E}[-e^{-\gamma((X+q)\tilde{\mu} - pq)}]$, where $p$ and $q$ are given implicitly by 
\begin{align}
\frac{a}{N-1}(S-X) + \xi - bp &= \frac{\tilde{u} - q}{N-1} \\
\alpha(X) - \beta(X)p &= q.
\end{align}
Rearranging we obtain 
\begin{align}
p &= F + C(S-X) - \lambda\tilde{u}  + \lambda q \\
q &= \bigg[\frac{\alpha(X)}{1+\lambda\beta(X)} - \frac{\beta(X)}{1+\lambda\beta(X)}\Big(F + C(S-X)\Big)\bigg] + \frac{\lambda\beta(X)}{1 + \lambda\beta(X)}\tilde{u},
\end{align}
where $F:= \frac{\xi}{b}$, $C := \frac{a}{b(N-1)}$ and $\lambda := \frac{1}{b(N-1)}$.

We note a few things about theses equations. Firstly, the parameters $F$, $C$ and $\lambda$ characterize the symmetric profile given by $a, \xi$ and $b$, and they are independent of the remaining agent's demand schedule. Secondly, equation (21) tells us that if we hold the remaining agent's initial inventory fixed, then the price the remaining agent receives is uniquely determined by the quantity she trades. Thus the agent is indifferent between demand schedules that lead to the same quantity of shares. Thirdly, from equation (22) we see that choosing a strategy $(\alpha, \beta)$ amounts to choosing functions $A:\mathbb{R} \to \mathbb{R}$ and $B:\mathbb{R} \to (0,1)$ such that the agent's traded quantity is $q = A(X) + B(X)\tilde{u}$. Fourthly, if the agent uses the linear strategy $(aX + \xi, b)$ then her traded quantity is $q = \frac{C}{\lambda}\big(X - \frac{S}{N}\big) + \frac{\tilde{u}}{N}$.

All this goes to show that linear symmetric equilibria correspond to $F, C \in \mathbb{R}$ and $\lambda > 0$ such that
\begin{equation}
\frac{C}{\lambda}\big(x - \frac{S}{N}\big) + \frac{\tilde{u}}{N} \in \argmax_{q \in V} \mathbb{E}\bigg[-e^{-\gamma\Big(x\tilde{\mu} + q\big(\tilde{\mu} - F - C(S-x) + \lambda\tilde{u}- \lambda q\big)\Big)}\bigg] \ \ \ \forall x \in \mathbb{R},
\end{equation}
where $V$ is as in (9). In equilibrium the price is
\begin{equation}
p = F + \frac{N-1}{N}CS- \frac{N-1}{N}\lambda\tilde{u},
\end{equation}
and agents' trades are
\begin{equation}
q^n = \frac{C}{\lambda}\big(X^n - \frac{S}{N}\big) + \frac{\tilde{u}}{N}.
\end{equation}

Now, note that $\tilde{u}$ is independent of $\tilde{\mu}$ and every $q \in V$ is $\tilde{u}$ measurable. Since the MGF of $\tilde{\mu}$ is $\mathbb{E}[e^{t\tilde{\mu}}] = e^{\mu t + \frac{\sigma^2}{2}t^2}$, we can compute the expectation in (23) for any $q \in V$ and $x \in \mathbb{R}$ as
\begin{equation}
e^{-\gamma\mu x} \mathbb{E}\bigg[-e^{-\gamma\Big(q\big(\mu - F - C(S-x) + \lambda\tilde{u}- \lambda q\big) - \frac{\gamma\sigma^2}{2}(x+q)^2\Big)}\bigg].
\end{equation}

For any $F, C, x \in \mathbb{R}$ and $\lambda > 0$, the function of $q$ in (26) is strictly concave over the convex set $V$. Thus the argmax in (23) consists of at most one point. Furthermore, for any $F, C, x \in \mathbb{R}$, $\lambda > 0$ and $u \in \mathbb{R}$, the much relaxed problem of maximizing the exponent inside the expectation in (26)
\begin{equation*}
\max_{q \in \mathbb{R}} q\big(\mu - F - C(S-x) + \lambda\tilde{u}- \lambda q\big) - \frac{\gamma\sigma^2}{2}(x+q)^2
\end{equation*}
has unique solution $\hat{q} = \frac{\mu - F - CS}{2\lambda + \gamma \sigma^2} + \frac{C - \gamma\sigma^2}{2\lambda + \gamma \sigma^2}x + \frac{\lambda}{2\lambda + \gamma \sigma^2}u$.

It follows that (23) holds if and only if
\begin{equation}
\frac{C}{\lambda}\big(x - \frac{S}{N}\big) + \frac{\tilde{u}}{N} = \frac{\mu - F - CS}{2\lambda + \gamma \sigma^2} + \frac{C - \gamma\sigma^2}{2\lambda + \gamma \sigma^2}x + \frac{\lambda}{2\lambda + \gamma \sigma^2}\tilde{u} \ \ \ \forall x \in \mathbb{R}.
\end{equation}
Since $\tilde{u}$ is non-degenerate, (27) holds if and only if $F = \mu, C = -\frac{\gamma\sigma^2}{N-1}$, and $\lambda = \frac{\gamma\sigma^2}{N-2}$. The theorem now follows by plugging these values is (24) and (25).
\end{proof}

\subsubsection{Discussion}
The proof is similar to the one in the previous subsection, with the idea being to characterize symmetric profiles in terms of the parameters $F$, $C$ and $\lambda$ of the induced optimal response problem. The additional parameter $C$ governs how the intercept of the supply curve in the optimal response problem depends on the optimizing agent's initial inventory. More specifically, $C$ captures the dependence of the intercept on the sum of all other agents' inventories, which the optimizing agent can compute by subtracting her own inventory from the total number of outstanding shares, i.e. $S - X$. 

The supply curve represents the prices at which the other agents are willing to clear the joint order of the liquidity traders and the optimizing agent. These prices must depend on the preexisting exposures of the remaining agents, hence the presence of the parameter $C$. For symmetric profiles, the others' exposure can be aggregated instead of considering individual exposures, which greatly simplifies the problem. 

The constant term in the equilibrium price is $\mu - \frac{\gamma\sigma^2}{N}S$, as opposed to just $\mu$ in the previous theorem. Thus there is a constant discount in the price reflecting the total number of outstanding shares. Intuitively this discount appears here because the market makers are already in possession of $S$ shares prior to the auction, whereas in the previous section they initially posses no shares. The coefficient of $\tilde{u}$ in the equilibrium price is $\frac{N-1}{N(N-2)}\gamma\sigma^2$, exactly as in the previous subsection. 

A unified way to write the equilibrium price in the two theorems is in terms of the aggregate inventory of the market makers \emph{after} the auction, denoted $S_{post}$. In the first subsection's model we have $S_{post} = \tilde{u}$ and in the second subsection we have $S_{post} = S + \tilde{u}$. In both cases, the equilibrium price is 
\begin{equation}
p = \mu - \frac{\gamma\sigma^2}{N}S_{post} - \frac{\gamma\sigma^2}{N(N-2)}\tilde{u}.
\end{equation}
The first two terms here are the competitive benchmark, and the last term is the deviation due to imperfect competition. 

At first glance it might seem surprising and counterintuitive that the deviation due to imperfect competition depends on $\tilde{u}$ and not $S_{post}$. For example, if $S_{post} >  0$ and $\tilde{u} < 0$, then the market makers are in aggregate long the asset, but the price is high relative to the competitive benchmark. The deviation due to imperfect competition should always favor the market makers, so one might expect it to make the price low when they are going long and high when they are going short. However this reasoning is flawed because the price in (28) is not the price at which the market makers enter their aggregate position of $S_{post}$. It is merely the price at which the market makers shift their aggregate position from $S$ to $S_{post}$. Said another way, (28) is not the denominator in the market makers' aggregate return, and as such the low/high long/short reasoning does not apply. 

The logic behind (28) is that the price is low when the liquidity traders are selling, $\tilde{u} > 0$, and high when they are buying, $\tilde{u} < 0$. Thus, roughly speaking, the liquidity traders are always "getting ripped off." This can be made more precise by recalling the analogy between the auction and a limit order book. Based on this analogy, the price in (17) can be interpreted as saying that the  is $\mu - \frac{\gamma\sigma^2}{N}S$, and orders walk the book at a rate of $\frac{\gamma\sigma^2}{N(N-2)}$ per share. 

Now, suppose the auction is repeated an instant later (prior to the realization of payoffs). Since the aggregate inventory of the market makers will be $S_{post}$ an instant later, the  will be $\mu - \frac{\gamma\sigma^2}{N}S_{post}$. So, (28) says that if the market makers buy in the first auction, $\tilde{u} > 0$, then they do so at a price lower than the  in the second auction. Similarly if they sell in the first auction, $\tilde{u} < 0$, then they do so at a price higher than the  in the second auction. Thus market makers always trade in the first auction at prices that are favorable relative to the  in the second auction. In particular, if a market maker were to unwind the position acquired in the first auction with a limit order\footnote{Technically a market order could also work in the second auction if it does not suffer too much price impact. The point is to unwind the position in the second auction at a favorable price relative to (28). Since the  in the second auction will be $\mu - \frac{\gamma\sigma^2}{N}S_{post}$, the price will always be favorable if using a limit order, and for a market order it depends on the price impact. In the dynamic model, equilibrium price impact will be constant over time and using market orders will not work for the market makers. However liquidity trader flow will mean revert, so the market makers will eventually be able to unwind their positions using limit orders, and this will earn positive profits.} in the second auction, then the roundtrip trade would earn positive profits. This is the sense in which the deviation due to imperfect competition always favors the market makers. This logic will reappear more clearly below, where we look at a dynamic model and the auction truly is repeated instant after instant. 

From (18) we see that the market makers trade heterogenous quantities of shares in the auction, unlike in the previous section where they all traded $\frac{\tilde{u}}{N}$. Since the market makers must clear $\tilde{u}$ in the auction, it follow that the average number of shares bought by each market maker is $\frac{\tilde{u}}{N}$, but some market makers might buy more and some less. The total initial inventory of the market makers is $S$, so the average inventory held by each market maker is $\frac{S}{N}$. (18) says that the market makers with above average inventories buy less, and those with below average inventories buy more. 

A more precise way to understand the traded quantities is in terms of the Pareto optimality of the inventory distribution before and after the auction. Since all market makers are identical, it would be Pareto optimal for them to hold $\frac{S}{N}$ shares before the auction and $\frac{S_{post}}{N}$ after the auction. Individual inventories after the auction are $X^n_{post} := X^n + q^n$ and from (18) it follows that 
\begin{equation}
X^n_{post} - \frac{S_{post}}{N} = \frac{1}{N-1}(X^n - \frac{S}{N}).
\end{equation}
Thus (29) says that trading in the auction moves inventories closer to efficiency by a factor of $\frac{1}{N-1}$. When $N$ takes its smallest value of $3$, the market makers only move halfway towards efficiency, whereas they move entirely towards efficiency in the limit as $N \to \infty$. Thus imperfect competition among the market makers results in imperfect risk sharing. 

\section{The Dynamic Model}

This section presents a continuous time model of trading. The model is set on an infinite horizon, and each instant in time consists of the auction from above. Traders hold cash and shares, and these fluctuate over time based on the outcomes of the auctions. As in Section 2 there are two types of traders: market makers, endogenous, and liquidity traders, exogenous. 

The first subsection formulates the model as a stochastic differential game, and the second subsection formally defines the equilibrium concept to be considered. The third subsection states the main theorem, which provides a complete closed-form characterization of linear symmetric equilibria. The proof of the theorem is given in Section 4, and an analysis of the equilibrium is given in Section 5. 

\subsection{The Model}
Fix a filtered probability space $(\Omega, \mathcal{F}, \{\mathcal{F}_t\}, \mathbb{P})$ equipped with two independent Brownian motions, $\{B^D_t\}$ and $\{B^{\tilde{S}}_t\}$, and satisfying the usual conditions. We consider a market on an infinite horizon where shares of a zero net supply risky asset are traded for cash. Cash is in infinitely elastic supply, earns no interest, and is the numeraire. The market is populated by two types of traders: $N \in \mathbb{N}$ market makers and a collection of liquidity traders. Each market maker starts out at time $0$ holding $X^n_0$ shares of the asset, $n =1, \cdots, N$. The liquidity traders start out with a collective shareholding of $-S_0$, which must satisfy $S_0 = X^1_0 + \cdots + X^N_0$ since the asset is in zero net supply. 

At each instant in time the traders transact with one another at a uniform price $p_t$. Trading occurs smoothly, meaning that each trader has a trading rate, which is the time derivative of her current shareholdings.\footnote{Shareholdings will also be referred to as inventories in what follows.} Trading rates and the trading price at each instant in time are determined by a demand schedule auction between the traders. Thus at time $t$ each trader submits an affine\footnote{It is actually not necessary to assume that agents can only submit affine schedules, as we will see below. For the class of equilibria we consider, when forming an optimal response an agent can achieve any trading rate via an affine demand schedule. Thus even if agents could submit arbitrary schedules, there would still be an equilibrium where they all submit affine schedules. Of course, there may also be other equilibria where agents submit more exotic schedules.} demand schedule of the form $q = u_t - v_t p$. This is a commitment to trade at rate $u_t - v_t p$ at time $t$ if the trading price is $p$. We denote this demand schedule by $(u_t, v_t)$. The process $\{(u_t, v_t)\}$ is required to be progressively measurable, though we will place more stringent conditions on it below.

The liquidity traders' collective demand schedule at time $t$ is assumed to be of the form $(-\mathcal{N}_t, 0)$. Thus the liquidity traders' demand schedule is a vertical line through the point $-\mathcal{N}_t$, interpreted as a market order to sell $\mathcal{N}_tdt$ shares over the time interval $[t, t + dt]$. The liquidity traders' collective inventory at time $t$ is denoted by $-S_t$, and it follows that $\{S_t\}$ evolves as a consequence of trading according to
\begin{equation}
dS_t = \mathcal{N}_tdt.
\end{equation}
Furthermore, the process $\{\mathcal{N}_t\}$ is exogenously specified according to
\begin{align}
\mathcal{N}_t &= -\phi(S_t - \tilde{S}_t) \\
d\tilde{S}_t &= -\psi\tilde{S}_tdt + \sigma_{\tilde{S}}dB^{\tilde{S}}_t
\end{align}
where $\phi, \psi, \sigma_{\tilde{S}} > 0$.  

The idea behind the dynamics (30)-(32) is that $-\tilde{S}_t$ is the current inventory target of the liquidity traders, and $\phi$ governs the speed with which they trade towards their target. Thus when their current inventory is above the target, $-S_t > -\tilde{S}_t$, the liquidity traders submit market orders to sell, $\mathcal{N}_t > 0$, and vice versa. If $\phi$ is large, the liquidity traders are impatient and submit large market orders to quickly reach their target. If $\phi$ is small, the liquidity traders are patient, submit smaller market orders, and only move slowly towards their target. One can think of $-\tilde{S}_t$ as the inventory the liquidity traders would want to hold if markets were perfectly liquid. Due to illiquidity they cannot instantly acquire this inventory, and instead do so gradually. Note that since $\{\tilde{S}_t\}$ mean reverts about zero, so too does $\{S_t\}$, i.e. we have $\mathbb{E}[S_T|S_t, \tilde{S_t}] \to 0$ as $T \to \infty$ almost surely $\forall t \geq 0$.

The market makers submit demand schedules in order to maximize certain objectives. Denote the market makers' trading rates by $q^n_t$, so that their inventories evolve according to
\begin{equation}
dX^n_t = q^n_tdt.
\end{equation}
If the market makers submit the demand schedule processes $\{(\frac{\alpha^n_t}{\beta^n_t}, \frac{1}{\beta^n_t})\}$, then their trading rates $q^n_t$ and the trading price $p_t$ are determined implicitly by 
\begin{align}
\alpha^n_t - \beta^n_t q^n_t &= p_t \ \ \ \ \ \forall n = 1, \cdots, N \\
q^1_t + \cdots + q^N_t &= \mathcal{N}_t.
\end{align}
Note that the price, trading rate, and inventory processes all depend on the choice of demand schedule processes. This dependence is suppressed in the notation. 

Each market maker has cash holding $M^n_t$ which evolve as a result of trading according to $dM^n_t = -q^n_tp_tdt$. The market makers are also assumed at time $t$ to have a common exogenous valuation of the asset as $D_t$, where
\begin{equation}
dD_t = \mu dt + \sigma_DdB^D_t
\end{equation}
and $\mu,\ \sigma_D > 0$ are fixed constants. Thus each marker maker values her book, consisting of joint holdings in cash and the asset, as $W^n_t = X^n_tD_t + M^n_t$. This is the market maker's wealth, computed by valuing shareholdings at $D_t$.

Each market maker chooses her demand schedule to maximize the objective
\begin{equation}
\mathbb{E}\bigg[\int_0^{\infty}e^{-\rho t}\Big(dW^n_t - \frac{\gamma}{2}d\langle W^n \rangle_t\Big)\bigg],
\end{equation}
where $\rho$, $\gamma > 0$. Here $\langle W^n \rangle_t$ is the quadratic variation of the market makers wealth, and so $d\langle W^n \rangle_t$ can be thought of as the variance of instantaneous wealth changes. Thus the integrand in (37) can be interpreted as a mean-variance utility flow from instantaneous returns, which means that (37) embodies myopic\footnote{Myopic because the utility flow comes from instantaneous returns.} mean-variance preferences over returns. 

One can compute the objective function in (37) as
\begin{equation}
\mathbb{E}\bigg[\int_0^{\infty}e^{-\rho t}\Big(-q^n_t(p_t - D_t) + \mu X^n_t - \frac{\gamma \sigma_D^2}{2}(X^n_t)^2\Big)dt\bigg].
\end{equation}
Thus the market makers want to buy, $q^n_t > 0$, when the price is below their valuation, $p_t - D_t < 0$, and vice versa. Furthermore, they enjoy holding inventory to the extent that valuations grow on average, $\mu > 0$, and they are averse to holding inventory to the extent that valuations are volatile, $\sigma_D^2 > 0$. 

Barring technicalities, this completes the description of the model as a stochastic differential game between the $N$ market makers. Indeed, the control process\footnote{We take as controls the parameters of the inverse demand as opposed to the demand. This simplifies much of the algebra below. We will continue to refer to the controls as demand schedules.} of each market maker is $\{(\alpha^n_t, \beta^n_t)\}$ and the controlled dynamics are described by (30) - (35). The coupled objectives of the market makers are given by (38), the coupling being induced by (34) and (35). The initial conditions for the game are the initial conditions for equations (32), (33), and (36). We are interested in studying the Nash equilibria of this game. 

What remains is to formally specify the equilibrium concept that we will consider for this game. This includes any measurability and integrability conditions that the admissible controls must satisfy, as well as any restrictions on the class of equilibria that will be studied. This is carried out in the next subsection.

\subsection{Equilibrium Concept}
We begin by specifying the admissibility conditions that the market makers' demand schedules must satisfy. Firstly, we need to make sure that the system (34) - (35) can be solved uniquely to define progressively measurable trading rates and prices. Secondly, prices and trading rates must be sufficiently well-behaved so that the (implicit) integrals in (33) and the double integral in (38) converge absolutely. Finally, we will want to place some measurability restrictions on the demand schedule processes in order to reflect the type of information that market makers have access to. This gives rise to the definition of admissible profiles of demand schedules.

\begin{defn}
Given initial conditions $(\vec{x}, \tilde{s}, d) \in \mathbb{R}^N \times \mathbb{R} \times \mathbb{R}$ for $(X^1_0, \cdots, X^N_0), \ \tilde{S}_0$, and $D_0$, we say that the profile of progressively measurable demand schedules $\{(\alpha^1_t, \beta^1_t)\}, \cdots, \{(\alpha^N_t, \beta^N_t)\}$ is \textbf{admissible starting from} $(\vec{x}, \tilde{s}, d)$ if:
\begin{enumerate}
\item $\beta^n_t > 0 \ \forall t \geq 0, \ \forall n = 1, \cdots, N$ almost surely 
\item $\int_0^T |q^n_t| dt < \infty \ \forall T \geq 0, \ \forall n = 1, \cdots, N$ almost surely
\item The double integral (38) converges absolutely
\item $\alpha^n_t, \beta^n_t \in \sigma\Big(\{D_s\}_{0 \leq s \leq t}, \{p_s\}_{0 \leq s < t}, \{X^n_s\}_{0 \leq s \leq t}, S_0\Big) \ \forall t \geq 0, \ \forall n = 1, \cdots N$.
\end{enumerate} 
\end{defn}

The first condition says that market makers may only submit strictly decreasing demand schedules, and it guarantees that the system (34) - (35) can be solved uniquely to define progressively measurable trading rates and prices. The second and third conditions ensure that the integrals in (33) and (38) converge. The fourth condition states that the information a market maker has access to at any moment in time consists of the history of valuations, the history of prices, the history of her own shareholdings, and the initial level of the liquidity traders' shareholdings. This is essentially the information that exchanges provide to traders in reality. 

We are interested in identifying admissible profiles of demand schedules that are Nash equilibria. Given a profile of demand schedules  $\{(\alpha^1_t, \beta^1_t)\}, \cdots, \{(\alpha^N_t, \beta^N_t)\}$ that's admissible starting from $(\vec{x}, \tilde{s}, d)$, denote by $J^n\big(\vec{x}, \tilde{s}, d, \{(\alpha^1_t, \beta^1_t)\}, \cdots, \{(\alpha^N_t, \beta^N_t)\}\big)$  the value of the double integral (38). This is the payoff that market maker $n$ receives when everyone's strategies are $\{(\alpha^1_t, \beta^1_t)\}, \cdots, \{(\alpha^N_t, \beta^N_t)\}$ and the initial conditions are $(X^1_0, \cdots, X^N_0) = \vec{x}, \ \tilde{S}_0 = \tilde{s}$, and $D_0 = d$.

\begin{defn}
Given initial conditions $(\vec{x}, \tilde{s}, d)$, we say that a profile of demand schedules  $\{(\alpha^1_t, \beta^1_t)\}, \cdots, \{(\alpha^N_t, \beta^N_t)\}$ is a \textbf{Nash Equilibrium starting from} $(\vec{x}, \tilde{s}, d)$ if:
\begin{enumerate}
\item The profile is admissible starting from $(\vec{x}, \tilde{s}, d)$
\item For any $n = 1, \cdots, N$, and for any demand schedule process $\{(\alpha_t, \beta_t)\}$ such that\\ $\{(\alpha^1_t, \beta^1_t)\}, \cdots, \{(\alpha^{n-1}_t, \beta^{n-1}_t)\}, \{(\alpha_t, \beta_t)\}, \{(\alpha^{n+1}_t, \beta^{n+1}_t)\}, \cdots, \{(\alpha^N_t, \beta^N_t)\}\big)$ is admissible starting from $(\vec{x}, \tilde{s}, d)$, we have that
\begin{align*}&J^n\big(\vec{x}, \tilde{s}, d, \{(\alpha^1_t, \beta^1_t)\}, \cdots, \{(\alpha^N_t, \beta^N_t)\}\big) \geq \\&J^n\big(\vec{x}, \tilde{s}, d, \{(\alpha^1_t, \beta^1_t)\}, \cdots, \{(\alpha^{n-1}_t, \beta^{n-1}_t)\}, \{(\alpha_t, \beta_t)\}, \{(\alpha^{n+1}_t, \beta^{n+1}_t)\},\\&\ \ \ \ \ \ \ \ \ \ \ \ \ \ \ \ \ \ \ \ \ \ \ \ \ \ \ \ \ \ \ \ \ \ \ \ \ \ \ \ \ \ \ \ \ \ \ \ \ \ \ \ \ \ \ \ \ \ \ \ \ \ \ \ \ \ \ \ \ \ \ \ \ \ \ \ \ \ \cdots, \{(\alpha^N_t, \beta^N_t)\}\big).\end{align*}
\end{enumerate}
\end{defn}

The second condition is the standard condition for a Nash equilibrium. It states that when all the market makers but market maker $n$ play their equilibrium demand schedules, the maximum payoff the $n^{th}$ market maker can earn is if she also plays her equilibrium demand schedule. In other words, when considering the optimal response problem against a profile of equilibrium demand schedules, each market maker finds it optimal to also use her equilibrium demand schedule. 

Unfortunately, identifying all the Nash equilibria in this model is intractable and beyond the scope of this paper. Instead we will focus on a special class of equilibria where all the market makers' demand schedules have a linear and symmetric structure. While this is fairly restrictive, the equilibria seem quite realistic and exhibit interesting dynamics. More specifically, we will only consider equilibria where all the market makers use demand schedules with the same constant slope and with an intercept that is the same linear function of individual state variables. The individual state variables will be $D_t, X^n_t$, and $S_t$. The precise formulation is given in the next definitions.  

\begin{defn}
A profile of demand schedules $\{(\alpha^1_t, \beta^1_t)\}, \cdots, \{(\alpha^N_t, \beta^N_t)\}$ is said to be \textbf{linear symmetric} if $\exists\ a, \lambda, b, c, \xi \in \mathbb{R}$ s.t.
\begin{align}
\alpha^n_t &= aX^n_t + bD_t + cS_t + \xi \\
\beta^n_t &= \lambda
\end{align}
$\forall t \geq 0, \ \forall n = 1, \cdots, N$. 
\end{defn}

\begin{defn}
We say that $a, \lambda, b, c, \xi \in \mathbb{R}$ are a \textbf{linear symmetric Nash equilibrium} if the linear symmetric profile defined by (39) and (40) is a Nash equilibrium starting from any set of initial conditions. 
\end{defn}

\subsection{Equilibrium Characterization}
\begin{thm}
Fix exogenous parameters $N \geq 3$, $\rho, \gamma, \sigma_D, \phi, \psi, \sigma_{\tilde{S}} > 0$ and $\mu \in \mathbb{R}$.
There is a unique linear symmetric Nash equilibrium with price
\begin{equation*}
p_t = D_t + \frac{\mu}{\rho} -\theta\frac{\gamma\sigma_D^2}{N}S_t - \frac{\gamma}{N}\frac{N-1}{N-2}\frac{\rho\sigma_D^2}{(\rho + \psi)(\rho + \phi)}\Big(\frac{1}{\delta} + \frac{1}{\rho}\Big)\mathcal{N}_t
\end{equation*}
and trading rates
\begin{equation*}
q^n_t = -\kappa\Big(X^n_t - \frac{S_t}{N}\Big) + \frac{1}{N}\mathcal{N}_t,
\end{equation*}
where 
\begin{equation*}
\kappa := \rho(N-2)\frac{\rho + \psi}{\rho + \delta}
\end{equation*}
\begin{equation*}
\delta := \sqrt{\rho^2 + 2(N-2)(\rho + \psi)(\rho + \phi)}
\end{equation*}
\begin{equation*}
\theta := \frac{(\rho + \psi + \phi)\delta - \psi\phi}{(\rho + \psi)(\rho + \phi)\delta}.
\end{equation*}
\end{thm}

\section{Proof of Theorem 3.5}

Before proceeding we prove the following lemma, which states the for linear symmetric profiles the admissibility\footnote{Admissibility for a linear symmetric profile means that (9) and (10) are admissible given \emph{any} initial conditions.} conditions manifest in simple constraints on the parameters $a$ and $\lambda$. One of the key points of the lemma is that under a linear symmetric profile market makers can infer $S_t$ from the history of prices and valuations. The market makers are thus able to implement the demand schedules (39) and (40) given their individual information sets.  

\begin{lem}
A linear symmetric profile is admissible if and only if $\lambda > 0$ and $\frac{a}{\lambda} < \frac{\rho}{2}$.
\end{lem}
\begin{proof}
Fix a linear symmetric profile given by $a, \lambda, b, c, \xi \in \mathbb{R}$ as in (39) and (40). We need to show that the four conditions for admissibility in Definition 1.1 are satisfied if and only if $\lambda > 0$ and $\frac{a}{\lambda} < \frac{\rho}{2}$. Clearly the first condition holds if and only if $\lambda > 0$. Next we will show that a profile satisfying (39) and (40) always satisfies the second and fourth conditions. Finally we will show that the third condition holds if and only if $\frac{2a}{\lambda} < \rho$, thus completing the proof of the lemma. 

Note that by combining equations (39) and (40) with equations (33) - (35) we can conclude that prices, trading rates, and inventories under a linear symmetric profile must satisfy
\begin{align}
p_t &= \Big(\frac{a}{N} + c\Big)S_t + bD_t +\xi - \frac{\lambda}{N}\mathcal{N}_t\\
X^n_t &= e^{\frac{a}{\lambda}t}\Big(X^n_0 - \frac{S_0}{N}\Big) + \frac{S_t}{N}\\
q^n_t &= \frac{a}{\lambda}\Big(X^n_t - \frac{S_t}{N}\Big) + \frac{1}{N}\mathcal{N}_t
\end{align}
$\forall t \geq 0$ and $\forall n = 1, \cdots, N$. These formulas imply that trading rates are almost surely continuous and thus the second condition holds.

To prove the fourth condition it suffices to prove that $S_t \in \sigma\Big(\{D_s\}_{0 \leq s \leq t}, \{p_s\}_{0 \leq s < t}, S_0\Big)$ $\forall t \geq 0$. Equations (30) and (41) imply that the following ODE holds path by path for the process $\{S_t\}$:
\begin{equation*}
\frac{dS_t}{dt} = \frac{N}{\lambda}\Big(\frac{a}{N} + c\Big)S_t + \mathcal{A}_t,
\end{equation*}
where the process $\{\mathcal{A}_t\}$ is defined by
\begin{equation*}
\mathcal{A}_t = \frac{N}{\lambda}(bD_t + \xi - p_t).
\end{equation*}
It's clear from this definition that $\{\mathcal{A}_s\}_{0 \leq s < t} \in \sigma\Big(\{D_s\}_{0 \leq s \leq t}, \{p_s\}_{0 \leq s < t}, S_0\Big) \ \forall t \geq 0$. Furthermore the ODE above implies that 
\begin{equation*}
S_t = e^{\frac{N}{\lambda}\Big(\frac{a}{N} + c\Big)t}S_0 + \int_0^te^{\frac{N}{\lambda}\Big(\frac{a}{N} + c\Big)(t-s)}\mathcal{A}_sds
\end{equation*}
from which it follows that $S_t \in \sigma\Big(\{D_s\}_{0 \leq s \leq t}, \{p_s\}_{0 \leq s < t}, S_0\Big) \ \forall t \geq 0$.

Turning to the third condition for admissibility, note that (31) and (41) - (43) imply that under a linear symmetric profile the integrand in (38) is of the form $e^{-\rho t}Q(X^n_t, D_t, \tilde{S}_t, S_t)$, where $Q$ is a second order polynomial with an $(X^n_t)^2$ coefficient of $- \frac{\gamma \sigma_D^2}{2} \neq 0$ . Thus $\exists$ constants $M_0, M_1, M_2, M_3, M_4 > 0$ such that 
\begin{equation*}
|Q(X^n_t, D_t, \mathcal{N}_t, S_t)| \leq M_0 + M_1\big((X^n_t)^2 + D_t^2 + \tilde{S}_t^2 + S_t^2\big)
\end{equation*}
\begin{equation*}
(X^n_t)^2 \leq M_2 + M_3|Q(X^n_t, D_t, \mathcal{N}_t, S_t)| + M_4(D_t^2 + \mathcal{N}_t^2 + S_t^2)
\end{equation*}
$\forall t \geq 0$ almost surely. Equations (30) - (32) and (36) imply that $e^{-\rho t}D_t^2$, $e^{-\rho t}\tilde{S}_t^2$, and $e^{-\rho t} S_t^2$ are integrable for any initial conditions. Thus from these bounds it follows that (38) converges absolutely for any initial conditions and for any $n = 1, \cdots, N$ if and only if $e^{-\rho t}(X^n_t)^2$ is integrable for any initial conditions and for any $n = 1, \cdots, N$. From (42) it follows that this latter condition holds if and only if $\frac{a}{\lambda} < \frac{\rho}{2}$.
\end{proof}

This section formulates the optimal response problem for an individual market maker as a standard stochastic control problem. A key point is that in forming an optimal response, a market maker can optimize directly over her trading rate. This is analogous to the idea in Section 2 of formulating the optimal response problem as maximization against a linear supply curve. 

\begin{prop}
$a, \lambda, b, c, \xi \in \mathbb{R}$ are a linear symmetric Nash equilibrium if and only if $\lambda > 0$, $\frac{a}{\lambda} < \frac{\rho}{2}$, and for any initial conditions we have that 
\begin{equation}
\frac{a}{\lambda}\Big(X_t - \frac{S_t}{N}\Big) - \frac{\phi}{N}(S_t - \tilde{S}_t) \in \argmax_{\{q_t\}} \mathbb{E}\bigg[\int_0^{\infty}e^{-\rho t}\Big(-q_t(p_t - D_t) + \mu X_t - \frac{\gamma \sigma_D^2}{2}(X_t)^2\Big)dt\bigg],
\end{equation}
where
\begin{align}
p_t &=\Big(\frac{a}{N-1} + c + \frac{\lambda\phi}{N-1}\Big)S_t + bD_t + \xi - \frac{\lambda\phi}{N-1}\tilde{S}_t - \frac{a}{N-1}X_t + \frac{\lambda}{N-1}q_t.
\end{align} 
The relevant dynamics are
\begin{align*}
dX_t &= q_tdt \\
dD_t &= \mu dt + \sigma_D dB^D_t \\
d\tilde{S}_t &= -\psi\tilde{S}_tdt + \sigma_{\tilde{S}}dB^{\tilde{S}}_t \\
dS_t &= -\phi(S_t - \tilde{S}_t)_tdt
\end{align*}
and the optimization is constrained to those processes $\{q_t\}$ such that
\begin{enumerate}
\item $\int_0^T|q_t|dt < \infty \ \ \ \ \ \forall T \geq 0$
\item The double integral on the right side of (44) converges absolutely 
\item $q_t \in \sigma(\{D_u\}_{0 \leq u \leq t}, \{S_u\}_{0 \leq u \leq t}, \{X_u\}_{0 \leq u \leq t},  \{\tilde{S}_u\}_{0 \leq u \leq t}) \ \ \ \ \ \forall t \geq 0$.
\end{enumerate}  
\end{prop}

\begin{proof}
To have a Nash equilibrium we must have a profile that satisfies the admissibility and optimality conditions in Definition 3.2. Lemma 4.1 states that for a linearly symmetric profile the admissibility condition holds if and only if $\lambda > 0$ and $\frac{a}{\lambda} < \frac{\rho}{2}$. Thus to prove the proposition we need only show that the optimality condition holds if and only if (44) holds. 

Consider the optimal response problem of an individual market maker when facing a linear symmetric profile of demand schedules given by $a, \lambda, b, c, \xi \in \mathbb{R}$. Denote the inventory process and trading rate process for the remaining market maker by $\{X_t\}$ and $\{q_t\}$. If the remaining market maker chooses the demand schedule process $\{(\alpha_t, \beta_t)\}$, then her trading rate and the price process $\{p_t\}$ are given implicitly by 
\begin{align}
p_t &=\Big(\frac{a}{N-1} + c\Big)S_t + bD_t + \xi -\frac{\lambda}{N-1}\mathcal{N}_t - \frac{a}{N-1}X_t + \frac{\lambda}{N-1}q_t\\
p_t &= \alpha_t - \beta_t q_t.
\end{align}

The optimal response problem for the remaining market maker is to choose the demand schedule process $\{(\alpha_t, \beta_t)\}$ such that the profile of all agents' schedules is admissible\footnote{In the rest of the proof we will say that the remaining market maker's demand schedule process is admissible if the corresponding profile of all market makers' demand schedules is admissible.}, and the objective in (44) is maximized over all such processes. In computing the objective the relevant equations are (46)\footnote{(45) and (46) are the same equations because of (31).}, (47), and the differential equations in the statement of the proposition. The optimality condition for Nash equilibrium is satisfied if and only if for any set of initial conditions a maximizing choice of demand schedule process is $(\alpha_t, \beta_t) = (aX_t + bD_t +cS_t + \xi, \lambda)$. We now make three points regarding this optimal response problem. 

The first point is that the remaining market maker will be indifferent between any two demand schedule processes leading to the same trading rate. Indeed for fixed initial conditions, the objective in (44) depends on the choice of demand schedule process only through the trading rate and price processes. Furthermore, from (46) we see that the price process depends on the demand schedule processes only through the trading rate process. Hence the objective is constant over demand schedule processes leading to the same trading rate process. 

The second point is that by choosing an appropriate admissible demand schedule, the remaining market maker can achieve any trading rate satisfying the three conditions in the proposition. To prove this it suffice to show that 
\begin{equation}
S_t, \ \mathcal{N}_t \in \sigma\Big(\{D_s\}_{0 \leq s \leq t}, \{p_s\}_{0 \leq s < t}, \{X_s\}_{0 \leq s \leq t}, \{q_s\}_{0 \leq s < t}, S_0\Big) \ \ \ \ \ \forall t\geq0
\end{equation}
for any choice of the remaining market maker's demand schedule. If this is the case, then for any $\{\tilde{q}_t\}$ satisfying the three conditions in the proposition we have that the demand schedule process
\begin{equation*}
\Bigg\{\bigg(\Big(\frac{a}{N-1} + c\Big)S_t + bD_t + \xi -\frac{\lambda}{N-1}\mathcal{N}_t - \frac{a}{N-1}X_t + \Big(1 + \frac{\lambda}{N-1}\Big)\tilde{q}_t,1\bigg)\Bigg\}
\end{equation*}
is admissible, and by (46) and (47) it gives the remaining market maker the trading rate process $\{q_t\} = \{\tilde{q}_t\}$. 

To prove (48) it suffices to prove that 
\begin{equation}
\{S_u\}_{0 \leq u \leq t} \in \sigma\Big(\{D_s\}_{0 \leq s \leq t}, \{p_s\}_{0 \leq s < t}, \{X_s\}_{0 \leq s \leq t}, \{q_s\}_{0 \leq s < t}, S_0\Big) \ \ \ \ \ \forall t\geq0
\end{equation}
for any choice of the remaining market maker's demand schedule. This is because by differentiating $\{S_u\}_{0 \leq u \leq t}$ we can recover $\{\mathcal{N}\}_{0 \leq u \leq t}$. To prove (49) note that (46) implies that for any choice of the remaining market maker's demand schedule we have  
\begin{align*}
\frac{dS_t}{dt} &= \frac{N-1}{\lambda}\Big(\frac{a}{N-1} + c\Big)S_t + \mathcal{A}_t \\
\mathcal{A}_t &= \frac{N-1}{\lambda}\Big(bD_t + \xi - \frac{a}{N-1}X_t + \frac{\lambda}{N-1}q_t - p_t\Big).
\end{align*}
From this (49) follows exactly as in Lemma 4.1.

The final point is that if the remaining market maker chooses the demand schedule process $(\alpha_t, \beta_t) = (aX_t + bD_t +cS_t + \xi, \lambda)$, then her trading rate is 
\begin{equation}
q_t = \frac{a}{\lambda}\Big(X_t - \frac{S_t}{N}\Big) + \frac{\mathcal{N}_t}{N}
\end{equation}
$\forall t \geq 0$. Indeed (50) follows simply by plugging $(\alpha_t, \beta_t) = (aX_t + bD_t +cS_t + \xi, \lambda)$ in (46) - (47) and solving for $q_t$.

This third point shows that if the remaining market maker follows the linear symmetric profile, then her trading rate is the maximizer in (44). The first two points imply that solving the optimal response problem against a linear symmetric profile is the same as solving the constrained optimization problem in the proposition. Hence the proposition follows. 
\end{proof}

This proposition associates to each linear symmetric profile $a, \lambda, b, c, \xi \in \mathbb{R}$ an optimization problem, as well as a candidate solution of the problem. The parameters provide an equilibrium if and only if the candidate is truly a solution. Below we will use the theory of stochastic control to derive first order conditions for the optimization problem. Using these we will demonstrate that there is a unique profile for which the candidate is a true solution, and thus there is a unique linear symmetric equilibrium. 

Before proceeding, we establish some notation relevant to the proposition and it's use below. Given a set of parameters $a, \lambda, b, c, \xi \in \mathbb{R}$, the optimization in the proposition is a standard stochastic control problem on an infinite horizon, with state space $(x, d, \tilde{s}, s) \in \mathbb{R}^4$ and control space $q \in \mathbb{R}$. Denote the covariables for the problem by $y = (y_x, y_d, y_{\tilde{s}}, y_s) \in \mathbb{R}^4$ and 
\begin{equation*}
z =   
\begin{pmatrix}
z_{xx} & z_{xd} & z_{x\tilde{s}} & z_{xs}\\
z_{dx} & z_{dd} & z_{d\tilde{s}} & z_{ds}\\ 
z_{\tilde{s} x} & z_{\tilde{s} d} & z_{\tilde{s}\tilde{s}} & z_{\tilde{s} s}\\
z_{sx} & z_{sd} & z_{s\tilde{s}} & z_{ss}
\end{pmatrix}
\in \mathbb{R}^{4\times4}.
\end{equation*}
The Hamiltonian for the problem is $H: \mathbb{R}^4 \times \mathbb{R}^4 \times \mathbb{R}^{4\times4} \times \mathbb{R} \to \mathbb{R}$ given by
\begin{multline*}
H(x, d, \tilde{s}, s, y, z, q) = qy_x + \mu y_d - \psi\tilde{s} y_{\tilde{s}} - \phi(s - \tilde{s}) y_s + \frac{\sigma_D^2}{2}z_{dd} + \frac{\sigma_{\tilde{S}}^2}{2}z_{\tilde{s}\tilde{s}} + \mu x - \frac{\gamma\sigma_D^2}{2}x^2 - q\big(P(x, d, \tilde{s}, s, q) - d\big).
\end{multline*}
where $P:\mathbb{R}^4 \times \mathbb{R} \to \mathbb{R}$ is the function specifying prices as a function of state and control from (45) above, i.e. 
\begin{equation*}
P(x, d, \tilde{s}, s, q) = \Big(\frac{a}{N-1} + c + \frac{\lambda\phi}{N-1}\Big)s + bd + \xi - \frac{\lambda\phi}{N-1}\tilde{s} - \frac{a}{N-1}x + \frac{\lambda}{N-1}q\\.
\end{equation*}
Also denote by $Q: \mathbb{R}^4 \to \mathbb{R}$ the mapping corresponding to the Markov control in (44), i.e. 
\begin{equation}
Q(x, d, \tilde{s}, s) = \frac{a}{\lambda}\big(x - \frac{s}{N}) - \frac{\phi}{N}(s - \tilde{s}).
\end{equation}
$H$ and $P$ describe the optimization problem, and $Q$ is the candidate solution. The functions $H, \ P$, and $Q$ all depend on the parameters a, $\lambda$, b, c, and $\xi$, but this dependence is suppressed in the notation.

We will prove that $a, \lambda, b, c, \xi \in \mathbb{R}$ are a linear symmetric Nash equilibrium if and only if 
\begin{align}
a &= -\frac{N-1}{\delta}\gamma\sigma_D^2 \\
\lambda &= \frac{N-1}{N-2}\frac{\rho\gamma\sigma_D^2}{(\rho + \psi)(\rho + \phi)}\Big(\frac{1}{\delta} + \frac{1}{\rho}\Big) \\
b &= 1\\
c &= -\frac{\rho(\rho + \psi + \phi)}{(\rho + \psi)(\rho + \phi)}\frac{\gamma\sigma_D^2}{N}\Big(\frac{1}{\delta} + \frac{1}{\rho}\Big) + \frac{\gamma\sigma_D^2}{\delta} \\
\xi &= \frac{\mu}{\rho}.
\end{align} prices and trading rates are given by (41) - (43), and plugging in these values gives the formulas in the statement of the theorem. 

We begin by proving the only if part of the statement. To this end, suppose the parameters $a, \lambda, b, c, \xi \in \mathbb{R}$ are a linear symmetric Nash equilibrium. Denote by $V(x, d, \tilde{s}, s)$ the value function of the optimization problem in Proposition 4.2, i.e. the value of the supremum in (44) when the initial conditions are $(X_0, D_0, \tilde{S}_0, S_0) = (x, d, \tilde{s}, s)$. 

Note that the value function is smooth. Indeed because the parameters provide an equilibrium, (44) must hold, and therefore $V(x, d, \tilde{s}, s)$ can be computed by evaluating the expectation in (44) along the control process $\frac{a}{\lambda}\Big(X_t - \frac{S_t}{N}\Big) + \frac{\phi}{N}(S_t - \tilde{S}_t)$. This provides us with an explicit expression for $V$, and by direct inspection it follows that $V$ is smooth. 

It follows that $V$ satisfies the following HJB equation:
\begin{equation}
\rho V(x, d, \tilde{s}, s) = \sup_{q \in \mathbb{R}} H(x, d, \tilde{s}, s, \nabla V, \nabla^2 V, q) \ \ \ \ \ \forall (x, d, \tilde{s}, s) \in \mathbb{R}^4.
\end{equation}
Furthermore, since $Q$ is an optimal Markov control, it also follows that 
\begin{equation}
Q(x, d, \tilde{s}, s) \in \argmax_{q \in \mathbb{R}} H(x, d, \tilde{s}, s, \nabla V, \nabla^2 V, q) \ \ \ \ \ \forall (x, d, \tilde{s}, s) \in \mathbb{R}^4.
\end{equation}
(57) is the classical result that if the value function is smooth then it satisfies the HJB equation \cite{touziOptimalStochasticControl2013}. When an optimal Markov control is known to exist, one way to prove (57) is to first prove (58) \cite{carmonaLecturesBSDEsStochastic2016}. A lemma explicitly proving (58) is included in the appendix to this paper for completeness. 

The first order condition for (58) is 
\begin{equation}
V_x(x, d, \tilde{s}, s) = P\big(x, d, \tilde{s}, s, Q(x, d, \tilde{s}, s)\big) - d + \frac{\lambda}{N-1}Q(x, d, \tilde{s}, s) \ \ \ \ \ \forall (x, d, \tilde{s}, s) \in \mathbb{R}^4.
\end{equation}
Anti-differentiating (59) it follows that $\exists$ a smooth function $w: \mathbb{R}^3 \to \mathbb{R}$ such that 
\begin{multline}
V(x, d, \tilde{s}, s) = \frac{1}{2}\frac{a}{N-1}x^2 + (b-1)xd - \Big(\frac{N-2}{N(N-1)}\lambda\Big)x\tilde{s} + \Big(\frac{N-2}{N(N-1)}a + c\Big)xs + \xi x + w(d, \tilde{s}, s)
\end{multline}
$\forall (x, d, \tilde{s}, s) \in \mathbb{R}^4$.

In summary, we've shown that if $a, \lambda, b, c, \xi \in \mathbb{R}$ are a linear symmetric Nash equilibrium then (57) - (60) hold. Combining these equations, we conclude that $\exists$ a smooth function $w: \mathbb{R}^3 \to \mathbb{R}$ such that
\begin{multline}
\frac{\rho}{2}\frac{a}{N-1}x^2 + \rho(b-1)xd -\rho\frac{N-2}{N(N-1)}\lambda x\tilde{s} + \rho\Big(\frac{N-2}{N(N-1)}a + c\Big)xs + \rho\xi x + \rho w(d, \tilde{s}, s) \\
= \Big(\frac{a^2}{\lambda(N-1)} - \frac{\gamma\sigma_D^2}{2}\Big)x^2 + \Big(\frac{a}{N-1} + \frac{N-2}{N(N-1)}\psi\lambda + c\Big)x\tilde{s} - \frac{2a^2}{N(N-1)\lambda}xs + \mu bx \\
+ \mu w_d - \psi\tilde{s} w_{\tilde{s}} + \tilde{s} w_s  + \frac{\sigma_D^2}{2}w_{dd} + \frac{\sigma_\mathcal{N}^2}{2}w_{\tilde{s}\tilde{s}} + \frac{\lambda}{(N-1)N^2}\Big(\tilde{s} - \frac{a}{\lambda}s\Big)^2
\end{multline}
$\forall (x, d, \tilde{s}, s) \in \mathbb{R}^4$.

It remains to be shown that (61) implies (52) - (56). Note that since the function $w$ is independent of $x$, the coefficients of $x^2,\ xd,\ x\tilde{s},\ xs$, and $x$ must be equal on the left and right hand sides of (32). Equating the coefficients of $xd$ and $x$ immediately gives (54) and (56). Equating the remaining coefficients yields the algebraic system 
\begin{align}
\frac{\rho}{2}\frac{a}{N-1} &= \frac{a^2}{\lambda(N-1)} - \frac{\gamma\sigma_D^2}{2} \\
-\rho\frac{N-2}{N(N-1)}\lambda &= \frac{a}{N-1} + \frac{N-2}{N(N-1)}(\psi+ \phi)\lambda + c\\
\rho\Big(\frac{N-2}{N(N-1)}(a + \lambda\phi) + c\Big) &= -\frac{2a^2}{N(N-1)\lambda} - \frac{N- 2}{N(N-1)}\phi^2\lambda -\phi \Big(\frac{a}{N-1} + c\Big).
\end{align}

Equation (63) gives $c$ in terms of $a$ and $\lambda$, from which it follow that if (52) and (53) hold then so too does (55). Next we plug this expression for $c$ into (64) and solve for $\lambda$ in terms of $a$ to get $\lambda = -\frac{1}{\rho + \psi}\Big(\frac{a}{N-2} - \frac{N-1}{N-2}\frac{\gamma\sigma_D^2}{\rho}\Big)$. From this equation it follows that if (52) holds then so too does (53). Hence it only remains to prove that (52) holds. Plugging this expression for $\lambda$ into (62), we see that $a$ must satisfy $a^2 = \frac{(N-1)^2\gamma^2\sigma_D^4}{\delta^2}$. Now, because the parameters provide an equilibrium, the constraint $\frac{2a}{\lambda} < \rho$ must hold. The positive root for $a$ violates the constraint and the negative root satisfies it, so it follows that (52) holds. 
 
We now prove the if part of the statement. To this end, suppose that $a,\ b,\ c,\ \xi$, and $\lambda$ are given by equations (52) - (56). We need to show (44) holds, i.e. that the mapping $Q$ in (51) provides an optimal Markov control for the stochastic control problem in Proposition 4.2. This will be done by by using the verification theorem for the HJB equation \cite{phamContinuoustimeStochasticControl2009}. 

Given a set of initial conditions, denote by $\{\hat{q}_t\}$ and $\{\hat{X}_t\}$ the trading rate and inventory processes arising from using the
Markov control given by $Q$, i.e. $d\hat{X}_t = Q(\hat{X}_t, D_t, \tilde{S}_t, S_t)dt$ and $\hat{q}_t = Q(\hat{X}_t, D_t, \tilde{S}_t, S_t)$. These processes depend on the choice of initial conditions, but this is suppressed in the notation. We need to show that $\exists$ a smooth function $V(x, d, \tilde{s}, s)$ such that
\begin{enumerate}
\item (57) holds
\item (58) holds
\item $e^{-\rho t}\mathbb{E}[V(\hat{X}_t, D_t, \tilde{S}_t, S_t)] \to 0$ as $t \to \infty$ for any choice of initial conditions
\item $\{\hat{q}_t\}$ satisfies the the constraints in Proposition 4.2 for any choice of initial conditions. 
\end{enumerate} 
If this can be done then it follows by the verification theorem that $Q$ is an optimal Markov control for the stochastic control problem in Proposition 4.2. 

Note that we can find a second order polynomial $w(d, \tilde{s}, s)$ that satisfies the equation 
\begin{equation*}
\rho w = \mu w_d - \psi\tilde{s} w_{\tilde{s}} - \phi(s - \tilde{s})w_s  + \frac{\sigma_D^2}{2}w_{dd} + \frac{\sigma_{\tilde{S}}^2}{2}w_{\tilde{s}\tilde{s}} + \Bigg(\frac{\phi}{N}\sqrt{\frac{\lambda}{N-1}}(s - \tilde{s}) + \frac{a}{N}\frac{1}{\sqrt{\lambda(N-1)}}s\Bigg)^2
\end{equation*}
on all of $\mathbb{R}^3$. Now define the function $V(x, d, \tilde{s}, s)$ by equation (60). Then the function $V$ is smooth and by construction equations (59) and (61) hold. Notice that as a function of $q$ the Hamiltonian is a quadratic polynomial with leading coefficient $-\frac{\lambda}{N-1}$. Since $\lambda > 0$ it follows that (59) is not only a necessary condition for (58) but also a sufficient one. Thus (58) holds. This implies that the equation (57) is precisely the equation (61), and so (57) also holds. 

Since $V$ is a second order polynomial, in order to prove the third condition it suffices to show that $\mathbb{E}[e^{-\rho t}\hat{X}_t^2]$, $\mathbb{E}[e^{-\rho t}D_t^2]$, $\mathbb{E}[e^{-\rho t}\tilde{S}_t^2]$, and $\mathbb{E}[e^{-\rho t}S_t^2]$ converge to $0$ as $t \to \infty$ for any set of initial conditions. This follows from the fact that $\frac{2a}{\lambda} < \rho$ and $\psi, \phi > 0$. 

Finally, we need to check that $\{\hat{q}_t\}$ satisfies the constraints in Proposition 3.1. The third constraint is trivial since $\hat{q}_t = Q(\hat{X}_t, D_t, \tilde{S}_t, S_t)$. This formula also implies $\{\hat{q}_t\}$ is almost surely continuous, and thus the first condition holds. Lastly, the second condition holds because $\frac{2a}{\lambda} < \rho$ and $\psi, \phi > 0$.

\section{Equilibrium Analysis}

\subsection{Price Analysis}
Consider the equilibrium price process given in Theorem 3.5. The four terms admit intuitive economic interpretations. The first term $D_t$ is simply the market makers' current valuation for the asset. The second term $\frac{\mu}{\rho}$ is a premium for expected valuation growth. $\mu$ is the drift of valuations, so on average valuations increase by $\mu (T-t)$ over a time interval $[t, T]$. As this is common knowledge among the market makers, this must be reflected in the price at time $t$. Otherwise, a profitable deviation from equilibrium would be to buy the asset at time $t$ and sell it at time $T$. Since market makers discount payoffs from time $T$ to time $t$ by $e^{-\rho(T-t)}$, the appropriate premium in the price to prevent this deviation is $\frac{\mu}{\rho}$. Said another way, we have\footnote{The notation here and below is $\mathbb{E}_t[\cdot] := \mathbb{E}[\cdot| \mathcal{F}_t]$.}
\begin{equation*}
\frac{\mu}{\rho} = \mathbb{E}_t\bigg[\int_t^\infty e^{-\rho(T-t)}dD_T\bigg].
\end{equation*}
Thus the second term is the (risk-neutral) present value of expected future changes in the asset's value.  

Since the market makers are not risk neutral, they also require risk compensations to take exposures to the asset. This is the role of the third term in the price $-\theta\frac{\gamma\sigma_D^2}{N}S_t$. When the market makers are in aggregate long the asset, so $S_t$ is positive, this term is negative and thus the asset is trading at a relatively low price. Because the asset is trading at a low price, market makers don't find it profitable to deviate from equilibrium by selling the asset to reduce their exposures. Similarly, when the market makers are in aggregate short the asset, this term causes the asset to trade at a high price, and thus market makers don't find it profitable to buy the asset to reduce their exposures. 

The magnitude of the compensation per unit of exposure is given by $\theta\frac{\gamma\sigma_D^2}{N}$. $\sigma_D^2$ is the volatility of valuations, so this is the amount of risk per unit of exposure to the asset. $\frac{\gamma}{N}$ is the aggregate risk aversion of the market makers, so this is the dollar compensation they require to hold a unit of risk. Thus $\frac{\gamma\sigma_D^2}{N}$ is the dollar compensation the market makers require to hold a unit of exposure to the asset. The intuitive role of the parameter $\theta$ is to take in to account fluctuations in future risk exposures and to discount them to the present. This will be made clearer in Proposition 5.4 below. 

The final term in the price process is $-\frac{\gamma}{N}\frac{N-1}{N-2}\frac{\rho\sigma_D^2}{(\rho + \psi)(\rho + \phi)}\Big(\frac{1}{\delta} + \frac{1}{\rho}\Big)\mathcal{N}_t$, which is the price impact the liquidity traders face on their trades, or equivalently the slope of the supply curve they face when trading. The liquidity traders submit market orders to trade at rate $-\mathcal{N}_t$, so they are buying when $\mathcal{N}_t < 0$ and selling when $\mathcal{N}_t > 0$. When the liquidity traders are buying, price impact causes the trading price to be high, and when the liquidity traders are selling, price impact causes the trading price to be low. This is the model's analogue of liquidity traders' market orders walking the book. As suggested by \cite{kyleInformedSpeculationImperfect1989}, we define liquidity in the model as the reciprocal of price impact and study it's comparative statics. 

\begin{defn}
$Price\ Impact := \frac{\gamma}{N}\frac{N-1}{N-2}\frac{\rho\sigma_D^2}{(\rho + \psi)(\rho + \phi)}\Big(\frac{1}{\delta} + \frac{1}{\rho}\Big)$
\end{defn}

\begin{defn}
$Liquidity := \frac{1}{Price\ Impact}$
\end{defn}

\begin{prop} 
\begin{enumerate}
\setlength{\itemsep}{3mm}
\item $\frac{\partial}{\partial \gamma}Liquidity < 0$. \\
Liquidity is decreasing in market makers' risk aversion.
\item $\frac{\partial}{\partial \sigma_D}Liquidity < 0$. \\
Liquidity is decreasing in fundamental volatility.
\item If $\frac{\gamma}{N}$ is held fixed then $\frac{\partial}{\partial N}Liquidity > 0$. \\
Liquidity is increasing in market maker competition. 
\item $\frac{\partial}{\partial \psi}Liquidity > 0$. \\ 
Liquidity is decreasing in order flow uncertainty.
\end{enumerate}
\end{prop}

To understand the fourth point, recall that the liquidity traders' order flow is driven by $\{\tilde{S}_t\}$, which is an Ornstein-Uhlenbeck process whose stationary distribution has variance $\frac{\sigma_{\mathcal{N}}^2}{2\psi}$. Thus as $\psi$ increases, the liquidity traders' orders arrive with less uncertainty. 

These comparative statics agree with real word intuition about liquidity, and thus they justify the theoretical definition of liquidity. Put simply, liquidity should reflect market makers' willingness to absorb temporary flow. Market makers are less willing to absorb temporary flow when they are more risk averse, when the asset is riskier, when there are not many of them, or when order flow is more uncertain. By the proposition, liquidity is also lower in the model in these situations. 

Next we study what happens to price impact in the competitive limit of the model. This is the limit when $N \to \infty$ and $\frac{\gamma}{N} \to \gamma_0$. Essentially one considers the sequence of models with $N$ market makers each having risk aversion $\gamma_N := N\gamma_0$. Thus for any model in the sequence, the aggregate risk aversion of the market makers is $\frac{\gamma_N}{N} = \gamma_0$. Hence the sequence considers increasingly competitive market making sectors that in aggregate have the same risk bearing capacity. Taking the limit of the sequence provides a perfectly competitive benchmark for the model.

\begin{prop}
In the competitive limit, we have that
\begin{equation*}
Price\ Impact \to \frac{\gamma_0\sigma_D^2}{(\rho + \phi)(\rho + \psi)}
\end{equation*}
and 
\begin{equation*}
p_t \to D_t + \frac{\mu}{\rho} - \mathbb{E}_t\bigg[\int_t^\infty e^{-\rho(T-t)}\gamma_0\sigma_D^2S_TdT\bigg] \ \ \ \ \ \forall (t, \omega) \in [0, \infty) \times \Omega
\end{equation*}
\end{prop}

That price impact does not vanish in the competitive limit is somewhat surprising. Naive intuition would suggest that market makers exercise their market power by charging price impact. So in the competitive limit, when individual market makers no longer have market power, price impact should vanish. The reason this intuition does not hold is given by the limiting value of the equilibrium price. 

In the competitive limit, market makers have an aggregate risk aversion of $\gamma_0 > 0$, so the equilibrium price should consist of a risk neutral present value as well as a discount based on the market makers' aggregate exposure. The first two terms in the limiting price, $D_t + \frac{\mu}{\rho}$, are a risk neutral present value as discussed above. Thus the third term $-\mathbb{E}_t\bigg[\int_t^\infty e^{-\rho(T-t)}\gamma_0\sigma_D^2S_TdT\bigg]$ should be the appropriate risk discount. 

As discussed above, $S_T$ is the aggregate exposure of the market makers at time $T$, and $-\gamma_0\sigma_D^2S_T$ is the dollar compensation they require in order to maintain this exposure. However, because $S_T$ evolves with $T$, this is only the exposure over the infinitesimal time interval $[T, T+dt]$. Furthermore, standing at time $t$, market makers require a compensation for the entire path of exposures they will be taking over $[t, \infty)$. The appropriate compensation for the exposure over $[T, T+dt]$ is $\gamma_0\sigma_D^2S_T$, and market makers discount payoffs from time $T$ to time $t$ by $e^{-\rho(T-t)}$, so the appropriate compensation for the entire path of exposures over $[t, \infty)$ is $-\mathbb{E}_t\bigg[\int_t^\infty e^{-\rho(T-t)}\gamma_0\sigma_D^2S_TdT\bigg]$. The presence of the parameter $\theta$ in Theorem 3.5 ensures that this expected present value arises in the limit. Thus $\theta$ should be thought of as inflating/deflating the standard risk premium so as to account for the expected future path of flow.

Said another way, the third term is exactly what is needed to prevent market makers deviating from equilibrium by pursuing a strategy that buys/sells based on the current level of the aggregate exposure. If $S_T$ does not evolve with $T$ and has a constant value equal to $S$, then this term reads $-\frac{1}{\rho}\gamma_0\sigma_D^2S$. This is exactly the risk discount a representative CARA investor with risk aversion $\gamma_0$ and time discount rate $\rho$ would require to hold $S$ shares of an asset with volatility $\sigma_D^2$. 

Based on this analysis, we can conclude, as suggested previously, that $-\frac{1}{\rho}\frac{\gamma\sigma_D^2}{N}S_t$ is not the full risk compensation component of the equilibrium price, but instead that $-\mathbb{E}_t\bigg[\int_t^\infty e^{-\rho(T-t)}\gamma_0\sigma_D^2S_TdT\bigg]$ is. Thus part of the price impact component of prices compensates market makers for risk, and so it does not vanish in the competitive limit. The rest of the price impact component of prices is a manifestation of market makers' market power, and so it does vanish in the competitive limit. The interpretation is that market makers charge price impact for (at least) two reasons. The first reason is simply because they can, since they have market power, and doing so is profitable. The second reason is that incoming trades change the entire path of exposures that market makers will be taking going forward. In order for the market makers to find it utility maximizing to clear the incoming trade, this must be reflected in the trading price.

\subsection{Inventory Analysis}
Next consider the equilibrium trading rates given in Theorem 3.5. The two terms encode two important properties of the market makers' equilibrium trading behavior. The first is that in aggregate the market makers must buy at rate $\mathcal{N}_t$, or equivalently the average trading rate of the market makers must be $\frac{1}{N}\mathcal{N}_t$. This is simply a consequence of market clearing, since the liquidity traders are selling at rate $\mathcal{N}_t$. That this is indeed the case is guaranteed by the second term in the formula for trading rates; the first terms all cancel out when aggregating/averaging. 

The first term in the trading rates dictates which market makers buy more/less than the average. Note that $\kappa > 0$, so the market makers with inventories larger than $\frac{S_t}{N}$ buy less, and vice versa. This brings us to the second important property of the market makers' trading behavior: they are continually moving towards a Pareto optimal allocation amongst themselves. The market makers must in aggregate hold $S_t$ shares, simply by market clearing. Since they are all identical, it would be Pareto optimal for everyone to hold $\frac{S_t}{N}$ shares. However, this certainly can't hold at time $0$, as initial inventories are exogenous and arbitrary. Beyond time $0$ efficiency of the allocations depends on endogenous trading behavior. 

Using the formula for trading rates, we can compute that in equilibrium the trajectory of each market maker's inventory is 
\begin{equation*}
X^n_t = e^{-\kappa t}(X^n_0 - \frac{S_0}{N}) + \frac{S_t}{N}.
\end{equation*}
Thus each market maker deviates from the efficient allocation by the first term. This term is non-zero if and only if $X^n_0 \neq \frac{S_0}{N}$, and in this case it converges monotonically towards $0$ over time. Thus the market makers take as given the inefficiency in their initial allocations, and then trade amongst themselves to make allocations more efficient. Allocations are inefficient at time $t$ is because they were inefficient at time $0$; the endogenous trading behavior of the market makers does not in any way create allocational inefficiencies.

But the market makers' trading behavior also does not perfectly remove the initial inefficiency in allocations. Indeed, inventories converge to efficiency at an exponential rate of $\kappa$, and this rate is not infinite. Based on the formula for $\kappa$, we see that the two main drivers of this rate are $\psi$, governing the uncertainty in order flow, and $N$, governing the degree of competition. As $\psi \to \infty$, so order flow uncertainty vanishes, or as $N \to \infty$, so the market is perfectly competitive, this rate of convergence goes to infinity. 

The story here is essentially one of individual market makers behaving strategically in an attempt to earn profits. Note that the price in the absence of liquidity trades ($\mathcal{N}_t = 0$), interpreted as the , is $D_t + \frac{\mu}{\rho} - \theta\frac{\gamma}{N}\sigma_D^2S_t$. Any market maker holding more than $\frac{S_t}{N}$ shares should find this price high relative to her exposure and should want to sell, and vice versa. However, market makers take into account price impact, so they know they can't actually trade to $\frac{S_t}{N}$ at this price. The market makers also know that at any moment a random liquidity trade might come in and move their exposure in the right direction, but now without price impact\footnote{More specifically, price impact \emph{favors} the market maker in this situation, as her limit order is getting hit by an incoming liquidity trader market order. If the market maker was insistent on moving to $\frac{S_t}{N}$, then she would have to place a market order, so price impact would go \emph{against} her.}. Thus instead of trading all the way to $\frac{S_t}{N}$, market makers only move partially and take the chance that the liquidity traders' flow will move them the rest of the way. In the absence of order flow uncertainty, or under perfect competition, the market makers no longer go through this calculation and simply trade all the way to $\frac{S_t}{N}$, converging instantly to efficiency (i.e. at an infinite rate). 

\section{Conclusion}
This paper presents a mathematical framework to model market making and to thereby generate an endogenous price impact function. The key insight is that price impact is the mechanism through which market makers earn profits while matching their books. The paper presents this idea in a stylized static model as well as in a rich continuous time model. The theory provides insights on the connection between market making and price impact, and as such is provides practical guidance on how to interpret measures of execution costs. 

\printbibliography[heading=bibintoc]

\section{Appendix: A Stochastic Control Lemma}
Fix a filtered probability space $(\Omega, \mathcal{F}, \{\mathcal{F}_t\}, \mathbb{P})$ equipped with a d-dimensional Brownian motion $\{B_t\}$ and satisfying the usual conditions. Consider a standard infinite horizon stochastic control problem with state space $\mathbb{R}^S$, control space $\mathbb{R}^A$, and randomness coming from $\{B_t\}$. Thus we take as given mappings $b: \mathbb{R}^S \times \mathbb{R}^A \to \mathbb{R}^S$, $\sigma: \mathbb{R}^S \times \mathbb{R}^A \to \mathbb{R}^{S \times d}$, and $f: \mathbb{R}^S \times \mathbb{R}^A \to \mathbb{R}$. We assume that $b$ and $\sigma$ are uniformly Lipschitz in their first variables and that $f$ is Borel measurable. Also fix a constant $\beta > 0$.

Denote by $\mathcal{A}$ the set of progressively measurable $\mathbb{R}^A$ valued processes $\alpha = \{\alpha_t\}$ such that 
\begin{equation*}
\mathbb{E}\bigg[\int_0^T |b(0, \alpha_t)^2 + |\sigma(0, \alpha_t)|^2dt\bigg] \ \ \ \ \ \forall T > 0.
\end{equation*}
Given $x \in \mathbb{R}^S$ and $\alpha = \{\alpha_t\} \in \mathcal{A}$, denote by $X^{x, \alpha} = \{X^{x, \alpha}_t\}$ the unique strong solution\footnote{The assumptions on $b$ and $\sigma$ and the definition of $\mathcal{A}$ were made precisely so that this equation admits a unique strong solution \cite{phamContinuoustimeStochasticControl2009}.} to the SDE
\begin{align*}
dX_t &= b(X_t, \alpha_t)dt + \sigma(X_t, \alpha_t)dB_t \\
X_0 &= x.
\end{align*}
For each $x \in \mathbb{R}^S$ denote by $\mathcal{A}(x)$ the subset of $\alpha = \{\alpha_t\} \in \mathcal{A}$ such that 
\begin{equation*}
\mathbb{E}\bigg[\int_0^\infty e^{-\beta t}|f(X^{x, \alpha}_t, \alpha_t)|dt\bigg] < \infty
\end{equation*}
and assume that $\mathcal{A}(x)$ is nonempty $\forall x \in \mathbb{R}^S$. 

For $x \in \mathbb{R}^S$ and $\alpha = \{\alpha_t\} \in \mathcal{A}(x)$ define the reward functional by
\begin{equation*}
J(x, \alpha) = \mathbb{E}\bigg[\int_0^\infty e^{-\beta t}f(X^{x, \alpha}_t, \alpha_t)dt\bigg].
\end{equation*}
The value function is defined for $x \in \mathbb{R}^S$ by 
\begin{equation*}
V(x) = \sup_{\alpha \in \mathcal{A}(x)} J(x, \alpha).
\end{equation*}
We also define the Hamiltonian $H: \mathbb{R}^S \times \mathbb{R}^S \times \mathbb{R}^{S \times S} \times \mathbb{R}^A \to \mathbb{R}$ by 
\begin{equation*}
H(x, y, z, a) = b(x, a) \cdot y + \frac{1}{2}trace\big(\sigma(x,a))\sigma^T(x,a)z\big) + f(x, a).
\end{equation*}

The significance of the Hamiltonian is that it describes how functions of a controlled state process evolve over time. That is, by Ito's formula, for any $\alpha = \{\alpha_t\} \in \mathcal{A}$, $h > 0$ and $g \in C^{1, 2}\big([0, \infty) \times \mathbb{R}^S\big)$ we have that 
\begin{align*}
g(t+h, X_{t+h}^{x, \alpha}) - g(t, X_t^{x, \alpha}) = &\int_t^{t+h} \frac{\partial g}{\partial t}(s, X_s^{x, \alpha}) + H\big(X_s^{x, \alpha}, \nabla g(s, X_s^{x, \alpha}), \nabla^2g(s, X_s^{x, \alpha}), \alpha_s\big) \\& - f(X_s^{x, \alpha}, \alpha_s)\ ds 
+ Martingale. 
\end{align*}

\begin{defn}
We say that $a: \mathbb{R}^S \to \mathbb{R}^A$ is an \textbf{optimal Markov control} if $\forall x \in \mathbb{R}^S$, $\exists \hat{\alpha}^x = \{\hat{\alpha}^x_t\} \in \argmax_{\alpha \in \mathcal{A}(x)} J(x, \alpha)$ such that $\hat{\alpha}^x_t = a\big(X^{x, \hat{\alpha}^x}_t\big)$ $\forall t \geq 0$ almost surely.  
\end{defn}

\begin{lem}
Suppose that $V \in C^2(\mathbb{R}^S)$ and $H$ is continuous. If $a: \mathbb{R}^S \to \mathbb{R}^A$ is continuous and provides an optimal Markov control then 
\begin{equation*}
a(x) \in \argmax_{\tilde{a} \in \mathbb{R}^A}  H(x, \nabla V(x), \nabla^2 V(x), \tilde{a}) \ \ \ \ \ \forall x \in \mathbb{R}^S.
\end{equation*}
\end{lem}
\begin{proof}
We proceed by contradiction. Suppose that the conclusion of the theorem does not hold. Then $\exists x_0 \in \mathbb{R}^S$, $\tilde{a}_0 \in \mathbb{R}^A$ and $\epsilon > 0$ such that 
\begin{equation*}
 H(x_0, \nabla V(x_0), \nabla^2 V(x_0), a(x_0)) < H(x_0, \nabla V(x_0), \nabla^2 V(x_0), \tilde{a}_0) - 4\epsilon.
\end{equation*}
By continuity of $a$, $H$ and $V$ it follows that $\exists$ a neighborhood $U_1$ of $x_0$ such that
\begin{equation}
H(x, \nabla V(x), \nabla^2 V(x), a(x)) < H(x, \nabla V(x), \nabla^2 V(x), \tilde{a}_0) - 3\epsilon \ \ \ \ \ \forall x \in U_1.
\end{equation}
Let $\{X_t\}$ be the unique strong solution to the SDE
\begin{align*}
dX_t &= b(X_t, a(X_t))dt + \sigma(X_t, a(X_t))dBt \\
X_0 &= x_0
\end{align*}
and let $\{\tilde{X}_t\}$ be the unique strong solution to the SDE 
\begin{align*}
d\tilde{X}_t &= b(\tilde{X}_t, \tilde{a}_0)dt + \sigma(\tilde{X}_t, \tilde{a}_0)dBt \\
\tilde{X}_0 &= x_0.
\end{align*}
Note that by smoothness of $V$ and continuity of $H$ we can find neighborhoods $U_2$ and $U_3$ of $x_0$ such that 
\begin{align}
|V(x) - V(y)| < \frac{\epsilon}{\beta} \ \ \ \ \ \forall x, y \in U_2 \\
| H(x, \nabla V(x), \nabla^2 V(x), \tilde{a}_0) -  H(y, \nabla V(y), \nabla^2 V(y), \tilde{a}_0)| < \epsilon \ \ \ \ \ \forall x, y \in U_3.
\end{align}
Define the stopping time $\tau = \inf\{t \geq 0: (X_t, \tilde{X}_t) \not\in U_1 \cap U_2 \cap U_3 \times U_1 \cap U_2 \cap U_3\}$. Since the processes $\{X_t\}$ and $\{\tilde{X}_t\}$ have a.s. continuous paths it follows that 
\begin{equation}
\tau > 0  
\end{equation}
almost surely. 

Now, we can estimate
\begin{align*}
V(x_0) &= \mathbb{E} \bigg[\int_0^\tau e^{-\beta t} f(X_t, a(X_t))dt + e^{-\beta\tau}V(X_\tau)\bigg] \\
&= V(x_0) + \mathbb{E} \bigg[\int_0^\tau e^{-\beta t}\Big(H\big(X_t, \nabla V(X_t), \nabla^2 V(X_t), a(X_t)\big) - \beta V(X_t)\Big)dt\bigg] \\
&\leq  V(x_0) + \mathbb{E} \bigg[\int_0^\tau e^{-\beta t}\Big(H\big(X_t, \nabla V(X_t), \nabla^2 V(X_t), \tilde{a}_0\big) - 3\epsilon - \beta V(X_t)
\Big)dt\bigg] \\
&\leq  V(x_0) + \mathbb{E} \bigg[\int_0^\tau e^{-\beta t}\Big(H\big(\tilde{X}_t, \nabla V(\tilde{X}_t), \nabla^2 V(\tilde{X}_t), \tilde{a}_0\big) + \epsilon - 3\epsilon - \beta V(\tilde{X}_t) + \epsilon\Big)dt\bigg] \\
&=  \mathbb{E} \bigg[\int_0^\tau e^{-\beta t} f(\tilde{X}_t, \tilde{a}_0)dt + e^{-\beta\tau}V(\tilde{X}_\tau)\bigg] - \frac{\epsilon}{\beta}\mathbb{E}[1 - e^{-\beta\tau}] \\
&\leq V(x_0) - \frac{\epsilon}{\beta}\mathbb{E}[1 - e^{-\beta\tau}].
\end{align*}
The first equality uses the dynamic programing principle \cite{phamContinuoustimeStochasticControl2009} and the optimality of the Markov control given by $a$. The second and last equality use Ito's lemma. The last inequality uses the dynamic programming principle. The inequalities in the middle follow from the definition of $\tau$ and inequalities (36) - (38). Since $\frac{\epsilon}{\beta} > 0$, it follows from this estimate that $\mathbb{E}[e^{-\beta\tau}] \geq 1$. However, this contradicts (39).
\end{proof}

\end{document}